\makeatletter\let\@twosidetrue\@twosidefalse\let\@mparswitchtrue\@mparswitchfalse\makeatother % center text on all pages (no alternating margins for even and odd pages)
\documentclass[a4paper,runningheads]{llncs}

\usepackage{etoolbox}
\newtoggle{arxiv}
\newcommand{\ifarxivelse}[2]{\iftoggle{arxiv}{#1}{\cite[#2]{arxivversion}}}
\toggletrue{arxiv}
\iftoggle{arxiv}{
\pagestyle{headings}
\usepackage[paperheight=235mm, paperwidth=155mm,textwidth=12.2cm,textheight=19.3cm]{geometry}
}{}

\usepackage[T1]{fontenc}
\usepackage{amsmath,amssymb}
\usepackage{xspace}
\usepackage{booktabs,multirow}
\usepackage{adjustbox}
\usepackage{cite}
\usepackage{wrapfig}
\usepackage{mathrsfs}
\usepackage{listings}
\usepackage{algorithm, algpseudocode}
\usepackage{graphicx}
\usepackage[%
  bookmarks,
  unicode,
  colorlinks=true,
  allcolors=blue!65!black!90,
  breaklinks=true]{hyperref}
\usepackage{bbding}

\usepackage{enumitem}
\setitemize{noitemsep,topsep=0pt,parsep=0pt,partopsep=0pt,leftmargin=12.0pt}
\setenumerate{noitemsep,topsep=0pt,parsep=0pt,partopsep=0pt}
\setdescription{noitemsep,topsep=2pt,parsep=0pt,partopsep=0pt,leftmargin=12.5pt}

\usepackage{chngcntr}

\makeatletter
\def\orcidID#1{\textsuperscript{\,\smash{\protect\raisebox{-1.25pt}{\href{http://orcid.org/#1}{\protect\includegraphics[scale=.8]{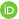}}}}}}
\makeatother

\allowdisplaybreaks % Used for the align*-environment to split equations over multiple pages

\usepackage[nameinlink,capitalize,english]{cleveref}
\Crefname{figure}{Fig.}{Figs.}
\crefname{figure}{fig.}{figs.}
\Crefname{tabular}{Tab.}{Tabs.}
\crefname{tabular}{tab.}{tabs.}
\Crefname{section}{Sec.}{Sects.}
\crefname{section}{sec.}{sects.}
\Crefname{appendix}{App.}{Apps.}
\crefname{appendix}{app.}{apps.}
\Crefname{equation}{Eq.}{Eqs.}
\crefname{equation}{eq.}{eqs.}
\creflabelformat{equation}{#2#1#3}
\Crefname{example}{Ex.}{Exs.}
\crefname{example}{ex.}{exs.}

\usepackage{tikz}
\usepackage{tikz-cd}
\tikzset{circled node/.style={circle,draw, inner sep=0, minimum size=1.5em},
nodestyle/.style={circled node,text height=.8em,text depth=.25em}}

\newcommand{\NN}{\ensuremath{\mathbb{N}}\xspace}  % natural  numbers
\newcommand{\RR}{\ensuremath{\mathbb{R}}\xspace}  % real     numbers
\newcommand{\PP}{\ensuremath{\mathcal{T}}\xspace}
\newcommand{\I}{\ensuremath{\mathscr{I}}\xspace}
\renewcommand{\O}{\ensuremath{\mathscr{O}}\xspace}
\newcommand{\U}{\ensuremath{\mathscr{U}}\xspace}
\newcommand{\V}{\ensuremath{\mathscr{V}}\xspace}
\newcommand{\1}{\ensuremath{\mathbf{1}}\xspace}
\newcommand{\tool}[1]{\textsc{#1}}
\newcommand{\lang}[1]{\textsc{#1}}

\newcommand{\toolset}{\tool{Modest Toolset}\xspace}
\newcommand{\eg}{e.g.\ }
\newcommand{\ie}{i.e.\ }
\newcommand{\etal}{et al.\xspace}
\newcommand{\wrt}{w.r.t.\xspace}

\newcommand{\DTMC}{\textup{DTMC}\xspace}

\newcommand{\SCC}{\textup{SCC}\xspace}

\newcommand{\set}[1]{\ensuremath{\{\,#1\,\}}}

\newcommand{\powerset}[1]{\ensuremath{2^{#1}}\xspace}
\newcommand{\Comp}[1]{\ensuremath{\mathrm{Comp}\left( #1\right)}\xspace}

\newcommand{\defeq}{\mathrel{\vbox{\offinterlineskip\ialign{\hfil##\hfil\cr{\tiny \rm def}\cr\noalign{\kern0.30ex}$=$\cr}}}}
\renewcommand{\P}{\ensuremath{\mathbf{P}}\xspace}
\newcommand{\M}{\ensuremath{\mathcal{M}}\xspace}

\makeatletter%
\g@addto@macro\normalsize{%
  \setlength\abovedisplayskip{3pt}%
  \setlength\belowdisplayskip{3pt}%
  \setlength\abovedisplayshortskip{-3pt}%
  \setlength\belowdisplayshortskip{3pt}%
}%
\makeatother

\begin{document}

\title{%
\iftoggle{arxiv}{DTMC Model Checking by Path\\ Abstraction Revisited (extended version)}
{DTMC Model Checking by\\ Path Abstraction Revisited}%
\thanks{
This work was supported
by the EU's Horizon 2020 research and innovation programme under MSCA grant agreement 101008233 (MISSION),
by the Interreg North Sea project STORM\_SAFE,
and
by NWO VIDI grant VI.Vidi.223.110 (TruSTy).
}
}
\iftoggle{arxiv}{\titlerunning{~}}{}

\author{%
Arnd Hartmanns\orcidID{0000-0003-3268-8674}
\and
Robert Modderman$^{\text{\,(\raisebox{-1.6pt}{\Envelope})}}$\orcidID{0009-0002-9198-3809}
}
\institute{%
University of Twente, Enschede, The Netherlands
$\cdot$ \email{r.modderman@utwente.nl}
}

\maketitle

\begin{abstract}
Computing the probability of reaching a set of goal states $G$ in a discrete-time Markov chain ({\DTMC}) is a core task of probabilistic model checking.
We can do so by directly computing the probability mass of the set of \emph{all} finite paths from the initial state to $G$; however, when refining counterexamples, it is also interesting to compute the probability mass of subsets of paths.
This can be achieved by splitting the computation into \emph{path abstractions} that calculate ``local'' reachability probabilities as shown by \'{A}brah\'{a}m \etal in 2010.
In this paper, we complete and extend their work:
We prove that splitting the computation into path abstractions indeed yields the same result as the direct approach, and that the splitting does not need to follow the \SCC structure.
In particular, we prove that path abstraction can be performed along \emph{any} finite sequence of sets of non-goal states.
Our proofs proceed in a novel way by interpreting the {\DTMC} as a structure on the free monoid on its state space, which makes them clean and concise.
Additionally, we provide a compact reference implementation of path abstraction in \lang{PARI/GP}.
\end{abstract}

\section{Introduction}

In this paper, we study methods to split computations of total reachability probabilities on discrete-time Markov chains ({\DTMC}s) into computing \emph{local} reachability probabilities:
Given a \DTMC with finite state space $S$ and initial state $a$, instead of computing the total probability $\P^a(\diamondsuit\,G)$ of reaching the set of goal states $G$ from $a$ (where all $s \in G$ are absorbing) by directly computing the probability mass of the set of paths from $a$ to $G$, we select a sequence of subsets of the set of non-goal states, compute local reachability probabilities \emph{over} those state subsets---\ie the probabilities of the paths within the subset that start from a transition entering the subset and end in a transition leaving the subset---and in the end compute the total probabilities from there.

Computing these local probabilities can be done by \emph{path abstraction}~\cite{AJWKB10}, which is an operation that, given a subset $S_1$ of the set of all non-goal states, moves the probability masses of the \emph{paths} that pass through $S_1$ onto new \emph{transitions} in a new \DTMC:
the path abstraction of the original \DTMC \emph{over} $S_1$.
Most states in $S_1$ become unreachable in the new \DTMC and could be removed.
%This can be seen as a local version of computing reachability probabilities since $S_1$ is a subset of the set of non-goal states.
The path abstraction operation can be used to split the computation of total reachability probabilities into multiple steps of computing local reachability probabilities on subsets of states (and thus of paths). % so that probability masses of subsets of the set of all paths from $a$ to $s$ are computed.
It also provides a recipe for probabilistic \emph{counterexample refinement}~\cite{AAR09}: % during this model checking procedure of computing total reachability probabilities:
each separate computation yields a candidate for the violation of a certain specified reachability constraint, e.g. those of the form $\P^a(\diamondsuit\,s)\leqslant\lambda$ for some goal state $s\in G$ and some given $\lambda\in[0,1]$.

\'{A}brah\'{a}m \etal in~\cite{AJWKB10} give a procedure to split the computation of total reachability probabilities by path abstraction based on a recursive decomposition into \emph{strongly connected components} ({\SCC}s) of the DTMC's underlying digraph up to trivial {\SCC}s. %, which is rooted in the structure of the underlying digraph of the \DTMC.
Furthermore, they give a concrete algorithm for computing path abstractions.
However, in \cite{AJWKB10}, no proof is supplied for the statement that the direct approach and the one that proceeds by splitting the computation over {\SCC}s yield the same result.
Furthermore, the algorithm as given in~\cite[Sec.~III]{AJWKB10} can be used only for subsets $S_1\subseteq S$ for which there is a path from every state within $S_1$ to a state outside of $S_1$, \ie for $S_1$ that do not contain bottom-SCCs (BSCCs), i.e., SCCs with no transitions leading outwards.
This can be ensured by additional preprocessing of the DTMC (\eg finding and collapsing all states that reach $s$ with probability zero~\cite{FKNP11}), but limits the algorithm's generality.

\paragraph{Our contribution.}
In \Cref{sec:path-abstr}, we prove that path abstraction is ``monotonically absorbing'' and thus both the direct ``global'' approach (of path abstracting a \DTMC over the set of all non-goal states straight away) and the ``local'' approach (of path abstracting along \emph{any} finite sequence of subsets of non-goal states and \emph{then} over the set of all non-goal states) yield the exact same result.
We achieve this in a novel, elegant way by interpreting a \DTMC as a structure on the free monoid on its state space (in \Cref{sec:Background}).
From our proof follows correctness of the approach of~\cite{AJWKB10}, which we show in \Cref{sec:DTMC-model-checking}, along with coupling our findings to counterexample refinement.
Furthermore, in \Cref{sec:num-methods}, we give a numerical algorithm to compute the path abstraction of a \DTMC over \emph{any} subset of the state space, taking into account that, when abstracting over a subset $S_1\subseteq S$, $S_1$ may contain states that do not reach outside of $S_1$.
Finally, in \Cref{sec:ref_impl}, we give a high-level reference implementation in the computer algebra system \lang{PARI/GP}~\cite{PARI2} to compute path abstractions, which in addition allows for the input of parametric {\DTMC}s.
This reference implementation is close to our abstract formulation of the algorithm, providing some confidence in its correctness.
It is a first step towards formalizing and machine-checking our algorithm and proofs using an interactive theorem prover, like recently done for the iterative interval iteration algorithm that computes reachability probabilities in a global manner~\cite{KSAHL25}.
%Investigating the precise time complexity improvements due to our results is also future work.
%We collect descriptions of future work in \Cref{sec:Conclusion}.

\paragraph{Related approaches.}
Computing reachability probabilities in DTMCs, mathematically, means solving a linear equation system~\cite[Sec.~10.1.1]{BK08}; research lies in doing so efficiently under different constraints and with a view towards different purposes.
Our work extends and completes that of \'{A}brah\'{a}m \etal~\cite{AJWKB10}, which uses \emph{path abstraction} over the SCC structure, with a view towards counterexample refinement.
We briefly compare path abstraction to \emph{state elimination}, which is prominently used for checking parametric DTMCs, in \Cref{sec:Conclusion}.
The idea of exploiting the SCC structure was already part of the model reduction techniques for Markov decision processes (MDPs), of which DTMCs are a special case, proposed by Ciesinski \etal~\cite{CBGK08} in 2008.
It also underpins the iterative-numeric \emph{topological value iteration} algorithm~\cite{DG07} as well as the incremental approach by Kwiatkowska \etal~\cite{KPQ11}, and improves parametric DTMC model checking~\cite{JCVWAKB14}.
Gui \etal~\cite{SGSLD13} then use the DTMC's SCC structure for computing reachability probabilities via Gaussian elimination, with a view towards improved scalability and performance by eliminating cycles.
Their work was later extended to MDPs~\cite{GSSLD14}.
The idea of reducing to acyclic DTMCs was in fact used earlier by Andr{\'{e}}s \etal~\cite{ADR08} for finding and describing probabilistic counterexamples.

\section{Background}
\label{sec:Background}

$\NN = \set{0,1,\ldots}$ is the set of natural numbers, and $\NN^+ \defeq \NN \setminus \set{0}$.
Given $n\in\NN^+$, let $[n]\defeq\set{1,\ldots,n}$. % denote the set of the first $n$ positive integers.
Given a set $X$, we write $\powerset{X}$ for the powerset of~$X$.
Given a propositional formula $\varphi$ and objects $A$ and $B$, let $[\varphi,A,B]$ be the object $A$ if $\varphi$ holds and $B$ otherwise.
In addition, let $[\varphi] \defeq [\varphi,1,0]$ (the \emph{Iverson bracket} of~$\varphi$).

%For a matrix $T$, let $T(i,j)$ denote the entry of $T$ in the $i$-th row and the $j$-th column.
We index matrices by finite sets, which is more general than indexing them by positive integers.
For a matrix $T$ indexed by $A \times B$, we let $T(a, b)$ be the entry of $T$ indexed by the pair $(a, b) \in A \times B$.
If $T$ is an $A \times B$ matrix, and $A'\subseteq A$ and $B'\subseteq B$, then let $T(A',B')$ denote the submatrix of $T$ whose rows and columns are those indexed by $A'$ and $B'$, respectively.
If $A'=B'$ then write $T(A')$ for $T(A',A')$.
Let $\1(A,B)$ denote the matrix with $\1(a,b)=[a=b]$.
If $A$ and $B$ are clear from the context, then we write $\1=\1(A,B)$.

\subsection{Combinatorics on Words: the Free Monoid}
\label{subsec:Combinatorics-on-Words}

Let us formalize and introduce some notions on free monoids on finite sets (or, equivalently, combinatorics on words), which is a field within mathematics that studies words and formal languages.
We borrow most of the terminology from~\cite{L02}.

Let $\Sigma$ be a finite set, the \emph{alphabet}, and let $\Sigma^*$ be the set of all finite sequences over $\Sigma$, the \emph{words} over $\Sigma$.
Given $x\in\Sigma^*$, let $|x|\in\NN$ be the \emph{length} of $x$; if in addition $x$ is non-empty, then let $x_i$ denote the $i$-th entry of $x$ for all $i\in[|x|]$.
Entries of words are called \emph{letters}.
The \emph{empty word} is denoted by~$\varepsilon$.
Given $n\in\NN$, let $\Sigma^{\Join n} \defeq \set{x\in\Sigma^*\mid|x|\Join n}$ where ${\Join}\in\set{<,\leqslant,\geqslant,>}$, and let $\Sigma^n \defeq \set{x\in\Sigma^*\mid|x|=n}$.
Whenever appropriate, we identify $\Sigma^1$ with $\Sigma$.
Furthermore, conforming to~\cite{L02}, we write $\Sigma^+=\Sigma^*\setminus\set{\varepsilon}$.

Given $x,y\in\Sigma^*$, let $xy$ denote the \emph{concatenation} of $x$ and $y$.
$\Sigma^*$ together with the concatenation operation and the empty sequence $\varepsilon$ forms a \emph{monoid}, as outlined in~\cite[Sec.~1.2.1]{L02}, known as the \emph{free monoid} on~$\Sigma$.
Concatenation extends to sets of words:
Given $X,Y\subseteq\Sigma^*$, we can form $XY \defeq \set{xy\mid x\in X\land y\in Y}$.
Furthermore, for $x,y\in\Sigma^*$ and $a\in\Sigma$, set $(xa)\star(ay) \defeq xay$.\footnote{In this fashion, we obtain a function $\star\colon\biguplus_{a\in\Sigma}\Sigma^*a\times a\Sigma^*\to\Sigma^+$.}
If $X,Y\subseteq\Sigma^*$ are such that there exists an $a\in\Sigma$ with $X\subseteq\Sigma^*a$ and $Y\subseteq a\Sigma^*$, then we may form $X\star Y \defeq \set{x\star y\mid x\in X\land y\in Y}$.

If $x\in\Sigma^*$, then $x'\in\Sigma^*$ is called a \emph{factor} of $x$ if there exist $y,z\in\Sigma^*$ such that $x=yx'z$, denoted by $x'\sqsubseteq x$.
If $y$ can be chosen empty, then $x'$ is called a \emph{prefix} of $x$, denoted by $x'\leqslant x$.
Note that $(\Sigma^*,\leqslant)$ in this fashion becomes a partially ordered set.
If $x'\leqslant x$ but $x'\neq x$, then this is denoted by $x'<x$.
If $z$ can be chosen empty, then $x'$ is called a \emph{suffix} of $x$.
Given $L\subseteq\Sigma^*$, we let $L^{\leqslant} \defeq \set{x\in L\mid x'<x\Rightarrow x'\notin L}$ denote the subset of $L$ of minimal elements \wrt $\leqslant$.\footnote{In~\cite{AJWKB10}, the prime symbol $'$ is used for this operation; % to denote this set of minimal elements with respect to the prefix order $\leqslant$;
we use $^\leqslant$ instead for readability.}
Note that the operator $(\cdot)^{\leqslant}$ on $\powerset{\Sigma^*}$ is idempotent, i.e., $(L^{\leqslant})^{\leqslant}=L^\leqslant$.

Now, given $\Sigma_1\subseteq\Sigma$,
we define the operation $-\Sigma_1\colon\Sigma^*\to\Sigma^*$ on $\Sigma^*$ as follows, where for ease of notation we write $x-\Sigma_1$ instead of $(-\Sigma_1)(x)$:
Given a word $x\in\Sigma^*$, we consider $x$ as the unique minimal-length\footnote{I.e., the letters of $x$ as a word over $\Sigma_1^+\uplus(\Sigma\backslash\Sigma_1)^+$ alternate over $\Sigma_1$ and $\Sigma\setminus\Sigma_1$.} word over the alphabet $\Sigma_1^+\uplus(\Sigma\setminus\Sigma_1)^+$ and let $x-\Sigma_1$ denote the word obtained from $x$ by replacing each letter $a\in\Sigma_1^+$ of $x$ by the first letter $a_1$ of $a$ when $a$ is considered as a word over $\Sigma$.
In other words, from every maximal factor $a\sqsubseteq x$ of $x$ with $a\in \Sigma_1^+$, we only keep its first entry.
To expand this, if $\Sigma_1,\ldots,\Sigma_n\subseteq\Sigma$, then we write $-(\Sigma_1,\ldots,\Sigma_n)$ for the operation $(-\Sigma_n)\circ\cdots\circ(-\Sigma_1)\colon\Sigma^*\to\Sigma^*$.
(Note that function composition order is read right-to-left.)
Furthermore, write, for $x\in\Sigma^*$, $x-(\Sigma_1,\ldots,\Sigma_n)$ instead of $(-(\Sigma_1,\ldots,\Sigma_n))(x)$.
Given $x\in\Sigma^*$, write $x+(\Sigma_1,\ldots,\Sigma_n)$ for the pre-image $(-(\Sigma_1,\ldots,\Sigma_n))^{-1}(x)=\set{x'\in\Sigma^*\mid x'-(\Sigma_1,\ldots,\Sigma_n)=x}$.
For $x+(\Sigma_1)$ we simply write $x+\Sigma_1$---thus, $x+\Sigma_1=\set{x'\in\Sigma^*\mid x'-\Sigma_1=x}$.

\begin{example}
Consider sequences over the Latin alphabet $\Sigma=\set{a,b,\ldots,z}$. We illustrate the $-\Sigma_1$-operation by a few cases of $x\in\Sigma^*$ and $\Sigma_1=\set{b,r,e,a,k}$.
\begin{enumerate}
    \item Let $x=error$.
    Then, $x$ as the unique minimal-length word over the alphabet $\Sigma_1^+\uplus(\Sigma\backslash\Sigma_1)^+$ is written as $x=(err)(o)(r)$.
    To obtain $x-\Sigma_1$, we replace the occurrence of $err$ by $e$ and the occurrence of $r$ by $r$, so $x-\Sigma_1=eor$.
    \item Since $x-\Sigma_1=eor$ we have $x\in eor+\Sigma_1$.
    \item For $x=spacebar$ we have $x-\Sigma_1=space$, since $x$ considered as the minimal-length word over $\Sigma_1^+\uplus(\Sigma\backslash\Sigma_1)^+$ is written as $x=(sp)(a)(c)(ebar)$.
    \item For $x=coffee$ we have $x+\Sigma_1=\varnothing$, as no $y\in\Sigma^*$ satisfies $y-\Sigma_1=x$ because $ee\sqsubseteq x$ is a factor of $x$ containing $e$ twice while $e\in\Sigma_1$.
\end{enumerate}
\end{example}

\subsection{Discrete-Time Markov Chains}
\label{subsec:DTMC}

Let us give a first formal definition of \emph{discrete-time Markov chains} ({\DTMC}s).
\begin{definition}%[Discrete-Time Markov Chain]
Let $S$ be a finite set and $a\in S$.
A \textbf{\emph{discrete-time Markov chain}} (DTMC) with \emph{state space} $S$ and \emph{initial state} $a\in S$ is a triple $(S,a,T)$ where $T\in\RR^{S\times S}$ is a \emph{stochastic (substochastic) matrix}, \ie $T(s,t)\geqslant0$ for all $s,t\in S$, and for all $s\in S$, $\sum_{t\in S}T(s,t)=1$ ($\leqslant1$).
\end{definition}
Since for the mechanism of path abstraction we need to reason about probabilities of paths and sets thereof, a convenient alternative interpretation of a \DTMC is as a structure on $S^+$. % where $S$ is its state space.
The key observation is that a (sub)stochastic matrix $T\in\RR^{S\times S}$ corresponds to a unique function $\P\colon S^+\to[0,1]$ such that $\P(x)=1$ for all $x\in S^1$, $\P(st)=T(s,t)$ for all $s,t\in S$, and $\P(xsy)=\P(xs)\P(sy)$ for all $x,y\in S^*$ and $s\in S$.

\begin{definition}%[Transition Probability Function]
\label{def:PMF}
Let $S$ be a finite set.
Then, let $\PP(S)$ denote the set of all functions $\P\colon S^+\to[0,1]$ satisfying
\begin{enumerate}
\item $\P(x)=1$ for all $x\in S^1$;
\item for all $s\in S$, $\sum_{t\in S}\P(st)\leqslant1$;
\item for all $x,y\in S^*$ and $s\in S$,
$\P(xsy)=\P(xs)\P(sy)$---or, equivalently, for all $s\in S$ and $u\in S^*s$ and $v\in sS^*$, $\P(u\star v)=\P(u)\P(v)$.
\end{enumerate}
Let $\PP_1(S) \defeq \set{\P\in\PP(S) \mid \forall s\in S\colon \sum_{t\in S}\P(st)=1}$.
Functions $\P\in\PP(S)$ are called \emph{substochastic transition probability functions} on $S^+$, and functions $\P\in\PP_1(S)$ are simply called \textbf{\emph{transition probability function}}\emph{s} on $S^+$.
\end{definition}
Note that any $\P\in\PP(S)$ can be extended to a function $\powerset{S^+}\to[0,\infty]$ by setting, for $R\subseteq S^+$, $\P(R)\defeq\sum_{x\in R^\leqslant}\P(x)$, which is guaranteed by~\cite[Theorem 0.0.2]{T11} to exist as either a nonnegative real number\footnote{By default, to the empty sum we assign the value $0$, so we obtain $\P(\varnothing)=0$.} or $+\infty$.
Note that we do not sum over $R$ but over $R^\leqslant$ instead. % (recall \Cref{subsec:Combinatorics-on-Words} for the definition of the $(\cdot)^{\leqslant}$-operator).
This way of defining probability masses of sets of paths is standard in the realm of probabilistic model checking on {\DTMC}s, see e.g.~\cite[Sec.~II]{AJWKB10}.
We can now give the following alternative definition of DTMCs, which we call the ``free monoid interpretation'':

\begin{definition}%[Discrete-Time Markov Chain---Free Monoid Interpretation]
\label{def:Markov}
Let $S$ be a finite set and let $a\in S$.
Then, a \textbf{\emph{discrete-time Markov chain}} (DTMC) with \emph{state space} $S$ and \emph{initial state} $a$ is a pair\footnote{The state space $S$ is implicitly given by the domain of $\P$ hence omitted.} $M=(\P,a)$ where $\P\in\PP_1(S)$.
Pairs $M=(\P,a)$ with $\P\in\PP(S)$ are called \emph{substochastic discrete-time Markov chains}.
We write $\M_1(S,a)$ and $\M(S,a)$ for the two classes of Markov chains, respectively.
\end{definition}
From now on, let $M=(\P,a)\in\M(S,a)$ always be given implicitly.
We are ultimately interested in computing \textbf{reachability probabilities}, \ie the probability $\P^a(\diamondsuit\,G)$ of reaching some set of goal states $G \subseteq S \setminus \set{a}$.
We write $\P^a(\diamondsuit\,s)$ for $\P^a(\diamondsuit\,\set{s})$.
W.l.o.g.\ we assume that each goal state is absorbing and $a$ is not.

\begin{definition}
\label{def:reaching-G}
Let $G \subseteq \set{ s \in S \mid \P(ss) = 1 }$ be a set of absorbing states with $a \notin G$.
Then the probability of reaching $G$ from $a$ is
$\P^a(\diamondsuit\,G) \defeq \P(aS^*G)$.
\end{definition}

\begin{example}
\Cref{subfig:DTMC_example} visualizes \DTMC $M_e \in \M_1(S_e, s_1)$ with $S_e=\set{s_1,\ldots,s_8}$. %with initial state $s_1$,
We have $\P(s_2) = 1$, $\P(s_2\,s_2) = 0$, $\P(s_2\,s_5) = \frac{1}{3}$, and $\P(s_2\,(s_5\,s_6)^+\,s_2) = \frac{1}{6}$.
Some reachability probabilities are $\P^{s_1}(\diamondsuit\,\set{s_7, s_8}) = 1$ and $\P^{s_1}(\diamondsuit\,s_7) = \frac{5}{9}$ (see~\Cref{ex:first-method}).
\end{example}

\begin{definition}%[Transitions, Underlying Digraph, Paths, Reachability]
\label{def:Markov-connectivity}
%Let $S$ be a finite set and let $a\in S$.
%Let $M=(\P,a)\in\M(S,a)$.
A \emph{transition} is a pair $(s,t)\in S\times S$ such that $\P(st)>0$.
The \emph{underlying digraph} of $M$ is the graph with vertex set $S$ whose edges are the transitions in~$M$.
A \emph{finite path} in $M$ is a %finite nonempty
sequence $x \in S^+$ such that $\P(x)>0$.
State $t$ is \emph{reachable} from state $s$ \emph{in} $M$ if $\P(sS^*t)>0$.
Given $K\subseteq S$, let $\Comp{M,K}$ denote the coarsest partition of $K$ such that each $C\in\Comp{M,K}$ satisfies $\P(sC^*t)>0$ for all distinct $s,t\in C$.
That is, $\Comp{M,K}$ are the \emph{strongly connected components} (SCCs) of the underlying digraph of $M$ restricted to $K$.
\end{definition}

\section{Path Abstraction via the Free Monoid}
\label{sec:path-abstr}

We define path abstraction using our new free monoid interpretation of {\DTMC}s, and in contrast to \cite{AJWKB10}, without altering the state space (we simply do not draw states that become isolated).
We then prove our main results on the monotonic absorption property of path abstraction that equates the global and local approaches to DTMC model checking.
For that, we need two \mbox{fundamental results}:

\begin{lemma}
\label{lem:sequence-abstr-is-monotonic-wrt-prefix-order}
Let $\Sigma$ be a finite set, and let $\Sigma_1\subseteq\Sigma$.
Then, for all $x\in\Sigma^*$, we have
$x'\leqslant x\Rightarrow x'-\Sigma_1\leqslant x-\Sigma_1$.
\end{lemma}

\begin{proof}[sketch]
By induction on $|x|-|x'|$.
The full proof is included in \ifarxivelse{\Cref{app:Proofs2}}{App.~A}.
\end{proof}

\begin{theorem}
\label{thm:words-combinatorics-main-results}
Let $\Sigma$ be a finite set.
Then, the following hold:
\begin{enumerate}
\item
Let $\Sigma_1\subseteq\Sigma$, let $x,y\in\Sigma^*$, and let $a\in\Sigma$.
Suppose that $xa\notin\Sigma^*\Sigma_1^2$.
Then, we have $xay-\Sigma_1=(xa-\Sigma_1)\star(ay-\Sigma_1)$.
\item
Again let $\Sigma_1\subseteq\Sigma$, let $x,y\in\Sigma^*$, and let $a\in\Sigma$.
Then, we have $(xa+\Sigma_1)^\leqslant\star(ay+\Sigma_1)^\leqslant=(xay+\Sigma_1)^\leqslant$.
\item
If $\Sigma_1\subseteq\Sigma_2\subseteq\Sigma$, then for all $x\in\Sigma^*$ we have $x-(\Sigma_1,\Sigma_2)=x-\Sigma_2$.
\item
Let $\Sigma_1\subseteq\Sigma_2\subseteq\Sigma$.
Then, for all $x\in\Sigma^*$ we have
$(x+\Sigma_2)^\leqslant=\biguplus\set{(y+\Sigma_1)^\leqslant\mid y\in(x+\Sigma_2)^\leqslant\land y-\Sigma_1=y}$.
\end{enumerate}
\end{theorem}

\begin{proof}[sketch]
Part 1: by case distinction $a\in\Sigma_1$ vs.\ $a\notin\Sigma_1$.
Part 2: by part 1.
Part 3: by considering $x$ as a word of minimal length over the alphabet $\Sigma_2^+\uplus(\Sigma\setminus\Sigma_2)^+$.
Part 4: by part 3 and Lemma \ref{lem:sequence-abstr-is-monotonic-wrt-prefix-order}.
The full proofs are included in \ifarxivelse{\Cref{app:Proofs2}}{App.~A}.
\end{proof}
We also need some initial results on transition probability functions.
We note that parts 2 and 3 below may appear trivial, but are later needed.

\begin{theorem}
\label{thm:PMF-main-results}
Let $S$ be a finite set, and let $\P\in\PP(S)$. Then,
\begin{enumerate}
\item for all $R\subseteq S^+$, we have $\P(R)=\lim_{k\to\infty}\P(R^\leqslant\cap S^{\leqslant k})$;
\item if $R\subset sS^*$ for some $s\in S$ with $R$ finite, then $\P(R)\leqslant1$;
\item if $R\subseteq sS^*$ for some $s\in S$ with $R$ of \emph{any} cardinality, then still $\P(R)\leqslant1$;
\item if $R\subseteq T\subseteq S^+$, then $\P(R)\leqslant\P(T)$.
\end{enumerate}
\end{theorem}

\begin{proof}[sketch]
Part 1: by~\cite[Theorem 0.0.2]{T11}.
Part 2: by induction on the maximum length of a sequence in $R^\leqslant$.
Part 3: by combining parts 1 and 2.
Part 4: via the function $f\colon R^\leqslant\to T^\leqslant$ where $f(x)$ is the unique nonempty prefix of $x$ in $T$ (hence is in $T^\leqslant$).
The full proofs are included in \ifarxivelse{\Cref{app:Proofs2}}{App.~A}.
\end{proof}

\subsection{Path Abstraction}

The path abstraction operation on (substochastic) {\DTMC}s intuitively collapses (all paths through) the abstraction set into (the paths crossing) the set's border states, while preserving the DTMC's overall reachability probabilities.

\begin{definition}%[Path Abstraction]
\label{def:path-abstr}
%Let $S$ be a finite set, and let $a\in S$.
For $S_1\subseteq S$, the \textbf{\emph{path abstraction}} of $M$ over $S_1$ is $M-S_1 \defeq (\P_{S_1}^a,a)\in\M(S,a)$\footnote{Formally, $M-S_1$ again denotes $(-S_1)(M)$ for a function $-S_1\colon\M(S,a)\to\M(S,a)$.} where $\P_{S_1}^a\colon S^+\to[0,\infty]$ is defined as follows:
%For $M=(\P,a)\in\M(S,a)$,
$$
\P_{S_1}^a(x) =
\begin{cases}
1 &\text{if }x\in S^1\\
0 &\text{if }x\in S^{\geqslant2} \wedge \exists\,i \colon x_i\in(S_1)_0^M\\
\P(x+S_1) &\text{otherwise}
\end{cases}
$$
with $(S_1)_0^M \defeq \set{s\in S_1\setminus\set{a}\mid\P((S\setminus S_1)s)=0}$.
%We claim that $\P_{S_1}^a\in\PP(S)$, proven in \Cref{lem:PKa-in-PMF(S)}, and set $M-S_1:=$.
\end{definition}
The set $(S_1)_0^M$ collects the states in the ``interior'' of $S_1$, \ie those with no incoming transition from outside $S_1$.
Note that $M-S_1$ has at most as many transitions as~$M$;
the second case removes all paths (and thus transitions) that pass through $(S_1)_0^M$.
For $S_1,\ldots,S_n\subseteq S$, %analogously to the setting $\Sigma^*\to \Sigma^*$ as outlined in \Cref{subsec:Combinatorics-on-Words},
we again write $-(S_1,\ldots,S_n)$ instead of $(-S_n)\circ\cdots\circ(-S_1)$, and $M-(S_1,\ldots,S_n)$ instead of $(-(S_1,\ldots,S_n))(M)$.

\begin{figure}[t]
\centering
\vspace{-1em}
\begin{minipage}[b]{0.49\textwidth}
\centering
\begin{tikzpicture}
\draw
(1.55,-0.4) node(S1){$S_1$}
(0.3, 1) node(ghost1){}
(1, 1) node[circled node](s1){$s_1$}
(2.4, 1) node[circled node](s2){$s_2$}
(1, 2.4) node[circled node](s3){$s_3$}
(2.4, 2.4) node[circled node](s4){$s_4$}
(2.4, -0.4) node[circled node](s5){$s_5$}
(3.8, -0.4) node[circled node](s6){$s_6$}
(3.8, 2.4) node[circled node](s7){$s_7$}
(3.8, 1) node[circled node](s8){$s_8$};
\draw
(4.7, 2.4) node{$1$}
(4.7, 1.0) node{$1$}
(4.0, 0.3) node{$\frac{1}{4}$}
(3.2, -0.9) node{$\frac{1}{2}$}
    (2.8, 0.0) node{$1$}
(3.0, 0.7) node{$\frac{1}{4}$}
(3.4, 1.8) node{$\frac{1}{12}$}
(3.1, 2.7) node{$\frac{1}{6}$}
(1.5, 2.9) node{$\frac{3}{4}$}
(2.0, 2.0) node{$1$}
(2.2, 0.3) node{$\frac{1}{3}$}
(1.5, 1.6) node{$\frac{2}{3}$}
(0.8, 1.7) node{$\frac{1}{6}$}
(1.6, 0.7) node{$\frac{5}{6}$};
\draw[-stealth] (ghost1) -- (s1);
\draw[-stealth] (s1) to[out = 0, in = 180] (s2);
\draw[-stealth] (s1) -- (s3);
\draw[-stealth] (s2) -- (s3);
\draw[-stealth] (s2) -- (s5);
\draw[-stealth] (s4) -- (s7);
\draw[-stealth] (s4) -- (s8);
\draw[-stealth] (s6) -- (s2);
\draw[-stealth] (s3) to[out = -30, in = -150] (s4);
\draw[-stealth] (s4) to[out = 150, in = 30] (s3);
\draw[-stealth] (s5) to[out = 30, in = 150] (s6);
\draw[-stealth] (s6) to[out = -150, in = -30] (s5);
\draw[-stealth] (s6) -- (s8);
\draw[-stealth] (s7) to[out = 30, in = -30, looseness = 8] (s7);
\draw[-stealth] (s8) to[out = 30, in = -30, looseness = 8] (s8);
\draw[densely dotted, line width=0.6pt] (2.4, 1.6)--(1.8,1.6);
\draw[densely dotted, line width=0.6pt] (1.8, 1.6)--(1.8,-1.2);
\draw[densely dotted, line width=0.6pt] (1.8, -1.2)--(4.7, -1.2);
\draw[densely dotted, line width=0.6pt] (4.7,-0.7)--(2.4,1.6);
\draw[densely dotted, line width=0.6pt] (4.7,-0.7)--(4.7,-1.2);
\end{tikzpicture}
\vspace{-4pt}
\caption{DTMC $M_e$ and set $S_1$}
\label{subfig:DTMC_example}
\end{minipage}\hfill
\begin{minipage}[b]{0.49\textwidth}
\centering
\begin{tikzpicture}
\draw
(0.3, 1) node(ghost1){}
(1, 1) node[circled node](s1){$s_1$}
(2.4, 1) node[circled node](s2){$s_2$}
(1, 2.4) node[circled node](s3){$s_3$}
(2.4, 2.4) node[circled node](s4){$s_4$}
(3.8, 2.4) node[circled node](s7){$s_7$}
(3.8, 1) node[circled node](s8){$s_8$};
\draw
(3.1, 0.7) node{$\frac{1}{5}$}
(4.7, 2.4) node{$1$}
(4.7, 1.0) node{$1$}
(3.4, 1.8) node{$\frac{1}{12}$}
(3.1, 2.7) node{$\frac{1}{6}$}
(1.5, 2.9) node{$\frac{3}{4}$}
(2.0, 2.0) node{$1$}
(1.5, 1.6) node{$\frac{4}{5}$}
(0.8, 1.7) node{$\frac{1}{6}$}
(1.6, 0.7) node{$\frac{5}{6}$};
\draw[-stealth] (ghost1) -- (s1);
\draw[-stealth] (s1) to[out = 0, in = 180] (s2);
\draw[-stealth] (s1) -- (s3);
\draw[-stealth] (s2) -- (s3);
\draw[-stealth] (s4) -- (s7);
\draw[-stealth] (s4) -- (s8);
\draw[-stealth] (s3) to[out = -30, in = -150] (s4);
\draw[-stealth] (s4) to[out = 150, in = 30] (s3);
\draw[-stealth] (s7) to[out = 30, in = -30, looseness = 8] (s7);
\draw[-stealth] (s8) to[out = 30, in = -30, looseness = 8] (s8);
\draw[-stealth] (s2)--(s8);
\draw (0.3,-1.1) node(ghost2){};
\end{tikzpicture}
\vspace{-4pt}
\caption{Abstracted \DTMC $M_e-S_1$}
\label{subfig:DTMC_example-s2s5s6}
\end{minipage}
%\vspace{-1em}
\end{figure}

\begin{example}
\label{ex:path-abstr}
The path abstraction $M_e - S_1$ of $M_e$ over $S_1=\set{s_2,s_5,s_6}$ is visualized in \Cref{subfig:DTMC_example-s2s5s6}.
In the abstraction from $M_e$ to $M-S_1$, the transitions between states from $S_1$ are removed, the states from $(S_1)_0^{M_e}=\set{s_5,s_6}$ are removed from the diagram (as their incoming and outgoing probabilities are set to zero), and the probabilities of the transitions $s_2\to s_3$ and $s_2\to s_8$ are replaced with the probability masses of the sets $s_2S_1^*s_3$ and $s_2S_1^*s_8$, respectively.
\end{example}
For \Cref{def:path-abstr} to be well-defined, we need to confirm that $(\P_{S_1}^a,a)\in\M(S,a)$:

\begin{lemma}
\label{lem:PKa-in-PMF(S)}
In \Cref{def:path-abstr}, we have $\P_{S_1}^a\in\PP(S)$.
\end{lemma}

\begin{proof}[sketch]
We show that $\P_{S_1}^a$ satisfies the three axioms of \Cref{def:PMF} as follows.
Axiom 1: by definition of $\P_{S_1}^a$, case~1.
Axiom 2: by \Cref{thm:PMF-main-results}, part~3.
Axiom 3: by \Cref{thm:words-combinatorics-main-results}, part~2.
The full proof is included in \ifarxivelse{\Cref{app:Proofs3}}{App.~B}.
\end{proof}
The reason why we consider substochastic {\DTMC}s (see \Cref{def:Markov}) is because there exist {\DTMC}s $M\in\M_1(S,a)$ and $S_1\subseteq S$ such that $M-S_1\notin\M_1(S,a)$, e.g. when $S_1$ contains a BSCC, and we want to keep the option open to abstract over sets that contain BSCCs.

\subsection{Monotonic Absorption of Path Abstraction}

We are now in position to prove the two main results of this paper, which state that path abstraction is ``monotonically absorbing'' in the sense that applying it first to subsets of later abstraction sets does not change the final result.\footnote{This could also be seen as a form of \emph{confluence}, but we prefer the term \emph{absorption} to emphasize the invariance of the result \wrt abstracting over \emph{subsets}.}

\begin{theorem}
\label{thm:path-abstr-main-results}
%Let $S$ be a finite set, and let $a\in S$.
If $S_1\subseteq S_2\subseteq S$, then $\forall\, M\in\M(S,a) \colon M-(S_1,S_2)=M-S_2$.
\end{theorem}

\begin{proof}[sketch]
By \Cref{thm:words-combinatorics-main-results}, part~3, \Cref{thm:words-combinatorics-main-results}, part~4, and \Cref{thm:PMF-main-results}, part~4.
The full proof is included in \ifarxivelse{\Cref{app:Proofs3}}{App.~B}.
\end{proof}
The preceding theorem provides structure in the analysis of path abstracting {\DTMC}s over (many) more than simply two subsets of the state space.

\begin{corollary}
\label{cor:path-abstr-refinement}
%Consider \Cref{thm:path-abstr-main-results}.
Let $S_1,\ldots,S_t\subseteq K \subseteq S$.
Then %, for every $M\in\M(S,a)$ we have
$M-(S_1,\ldots,S_t,K)=M-K$.
\end{corollary}

\begin{proof}[sketch]
By induction on $t$ and by \Cref{thm:path-abstr-main-results};
the full proof is in \ifarxivelse{\Cref{app:Proofs3}}{App.~B}.
\end{proof}

\section{DTMC Model Checking by Path Abstraction}
\label{sec:DTMC-model-checking}

The main goal of \DTMC model checking by path abstraction is to obtain the \emph{result} of abstracting a \DTMC $M$ over the set $K \defeq \set{s\in S \mid \P(ss)<1}$ of its non-absorbing states, \ie the DTMC $M-K$.
%For the remainder of this section, fix $K$ as such.
%In this fashion, $M$ is transformed into $M-K$.
Clearly, this preserves the self-loops with probability $1$ of all absorbing states, and the only other transitions %non-loop transitions (of positive probability)
remaining are the transitions $a\rightarrow s$ with $s\in S\setminus K$ absorbing.
They carry precisely the reachability probabilities $\P_K^a(as)=\P(aK^*s)=\P(aS^*s)=\P^a(\diamondsuit\,s)$.\footnote{$\P(aS^*s)=\P^a(\diamondsuit\,s)$ is by \Cref{def:reaching-G}.
The fact that $\P(aK^*s)=\P(aS^*s)$ is seen as follows:
We have $(aS^*s)^\leqslant=a(S\setminus\set{s})^*s$ and $(aK^*s)^\leqslant=aK^*s$, so any $x\in(aS^*s)^\leqslant\setminus(aK^*s)^\leqslant$ must contain a state $x_i\neq s$ that is absorbing, hence, as $x$ ends in $s$ and $x_i$ has outgoing probability zero, has probability zero.}
Then, for any goal set $G\subseteq S\setminus K$ we can simply compute $\P^a(\diamondsuit\,G)=\sum_{s\in G}\P^a(\diamondsuit\,s)$.

%\paragraph{Strategy to establish correctness of methods in \cite{AJWKB10}.}
\'{A}brah\'{a}m \etal in~\cite{AJWKB10} give two methods to perform DTMC model checking by path abstraction:
The first method just abstracts over each SCC (leaving the algorithm to abstract, \ie solve, each SCC open), while the second method also recursively applies path abstraction within SCCs.
We argue that both methods implicitly generate a finite sequence $S_1,\ldots, S_t\subseteq K$ of subsets of $K$ along which the original \DTMC is abstracted, before it is eventually abstracted over~$K$.
This means that they compute $M-(S_1,\ldots,S_t,K)$, which by \Cref{cor:path-abstr-refinement} is precisely $M-K$:
the desired output (substochastic) \DTMC.
Let us fix $S^{(0)}=\set{\set{s} \mid s\in S\land \P(ss)=0}$ for the set of singleton sets of states that do not have self-loops.

\paragraph{SCC abstraction.}
The first method, given by~\cite[Algorithm 1 alone]{AJWKB10}, is functionally equivalent to transforming $M$ into $M-(U_1,\ldots,U_t, K)$ where $U_1,\ldots,U_t$ is any enumeration of $\Comp{M,K}\setminus S^{(0)}$.
Since every $U_i$ is a subset of $K$, by \Cref{cor:path-abstr-refinement} this in turn is functionally equivalent to transforming $M$ into $M-K$. %, which is the desired output \DTMC.

\begin{figure}[t]
    \centering
    \vspace{-1em}
    \begin{minipage}[b]{0.49\textwidth}
    \centering
    \begin{tikzpicture}
    \draw
    (0.3, 1) node(ghost1){}
    (1, 1) node[circled node](s1){$s_1$}
    (2.4, 1) node[circled node](s2){$s_2$}
    (2.4, 2.4) node[circled node](s3){$s_3$}
    (3.8, 2.4) node[circled node](s7){$s_7$}
    (3.8, 1) node[circled node](s8){$s_8$};
    \draw
    (3.1, 0.7) node{$\frac{1}{5}$}
    (4.7, 2.4) node{$1$}
    (4.7, 1.0) node{$1$}
    (3.1, 2.7) node{$\frac{2}{3}$}
    (3.3, 1.9) node{$\frac{1}{3}$}
    (1.5, 1.9) node{$\frac{1}{6}$}
    (2.2, 1.7) node{$\frac{4}{5}$}
    (1.6, 0.7) node{$\frac{5}{6}$};
    \draw[-stealth] (ghost1) -- (s1);
    \draw[-stealth] (s1) to[out = 0, in = 180] (s2);
    \draw[-stealth] (s1) -- (s3);
    \draw[-stealth] (s2) -- (s3);
    \draw[-stealth] (s7) to[out = 30, in = -30, looseness = 8] (s7);
    \draw[-stealth] (s8) to[out = 30, in = -30, looseness = 8] (s8);
    \draw[-stealth] (s2)--(s8);
    \draw[-stealth] (s3)--(s8);
    \draw[-stealth] (s3)--(s7);
\end{tikzpicture}
\caption{\DTMC $M_e-(S_1,S_2)$}
\label{subfig:DTMC_example-s2s5s6-s3s4}
\end{minipage}\hfill
\begin{minipage}[b]{0.49\textwidth}
    \centering
    \begin{tikzpicture}
    \draw
    (0.3, 1) node(ghost1){}
    (1, 1) node[circled node](s1){$s_1$}
    (2.4, 2.4) node[circled node](s7){$s_7$}
    (2.4, 1) node[circled node](s8){$s_8$};
    \draw
    (3.3, 2.4) node{$1$}
    (3.3, 1.0) node{$1$}
    (1.5, 1.9) node{$\frac{5}{9}$}
    (1.6, 0.7) node{$\frac{4}{9}$};
    \draw[-stealth] (ghost1) -- (s1);
    \draw[-stealth] (s1) to[out = 0, in = 180] (s2);
    \draw[-stealth] (s1) -- (s3);
    \draw[-stealth] (s7) to[out = 30, in = -30, looseness = 8] (s7);
    \draw[-stealth] (s8) to[out = 30, in = -30, looseness = 8] (s8);
\end{tikzpicture}
\caption{The final \DTMC $M_e-K$}
\label{subfig:DTMC_example-K}
\end{minipage}
%\vspace{-1em}
\end{figure}

\begin{example}
\label{ex:first-method}
If we apply the first method to $M_e$ from \Cref{subfig:DTMC_example}, then we abstract over $S_1=\set{s_2,s_5,s_6}$, $S_2=\set{s_3,s_4}$, and finally $K=\set{s_1,\ldots,s_6}$.
We obtain the sequence $M_e$, $M_e-S_1$, $M_e-(S_1,S_2)$, and $M_e-(S_1,S_2,K)=M_e-K$, depicted in Figs.~\ref{subfig:DTMC_example}, \ref{subfig:DTMC_example-s2s5s6}, \ref{subfig:DTMC_example-s2s5s6-s3s4}, and \ref{subfig:DTMC_example-K}, respectively.
From \Cref{subfig:DTMC_example-K} it is clear that $\P^{s_1}(\diamondsuit\,s_7)=\frac{5}{9}$ and $\P^{s_1}(\diamondsuit\,s_8)=\frac{4}{9}$, where $\P$ is the transition probability function of the first \DTMC $M_e$ as depicted in \Cref{subfig:DTMC_example}, but, due to \Cref{cor:path-abstr-refinement}, in expressing the reachability probabilities of $\diamondsuit\,s_7$ and $\diamondsuit\,s_8$, can be replaced with the transition probability function of \emph{any} of the abstracted {\DTMC}s we computed along the way.
\tikzset{circled node/.style={circle,draw, inner sep=0, minimum size=1.5em},
nodestyle/.style={circled node,text height=.8em,text depth=.25em}}
\end{example}

\paragraph{Recursive abstraction.}
The second method, described by \cite[Algorithm 1 using Algorithm~2]{AJWKB10}, uses a recursive approach: % to path abstraction.
To abstract an SCC $S_1$, it abstracts each SCC of $(S_1)_0^M$ (which is the interior of $S_1$ from \Cref{def:path-abstr}), and so on.
We capture this by a new operator~$\ominus$:
$$
M\ominus S_1 := (M\ominus(U_1,\ldots,U_t))-S_1
$$
where
(i)~$S_1 \subseteq S$ is strongly connected (\ie $\P(sS_1^*t)>0$ for all distinct $s,t\in S_1$),
(ii)~$(S_1)_0^M\subsetneq S_1$, and
(iii)~$U_1,\ldots,U_t$ is an enumeration of $\Comp{M,(S_1)_0^M}\setminus S^{(0)}$.
If $U_1,\ldots,U_t$ in condition (iii) is empty, $\ominus$ returns $M - S_1$ as a base case.

At some point, $\ominus$ will be applied to an $S_1$ for which $(S_1)_0^M$ does not contain non-trivial {\SCC}s, and then regular ``$-$''-abstractions are computed bottom-up.\footnote{Condition (ii) is required for the initial $S_1$ to ensure termination, too; see \ifarxivelse{\Cref{app:differences-with-AJWKB10}}{App.~C}.}
In this way, the input \DTMC $M$ is, again, ``$-$''-abstracted along a sequence of subsets of $K$ and finally over $K$ itself, which by \Cref{cor:path-abstr-refinement} yields $M-K$ again. %, which is, again, the desired output \DTMC.

\begin{figure}[t]
    \centering
    %\vspace{-1em}
\begin{minipage}[b]{0.49\textwidth}
\centering
\begin{tikzpicture}
    \draw
    (0.3, 1) node(ghost1){}
    (1, 1) node[circled node](s1){$s_1$}
    (2.4, 1) node[circled node](s2){$s_2$}
    (1, 2.4) node[circled node](s3){$s_3$}
    (2.4, 2.4) node[circled node](s4){$s_4$}
    (3.8, -0.4) node[circled node](s5){$s_5$}
    (3.8, 2.4) node[circled node](s7){$s_7$}
    (3.8, 1) node[circled node](s8){$s_8$};
    \draw
    (4.7, 2.4) node{$1$}
    (4.7, 1.0) node{$1$}
    (4.0, 0.3) node{$\frac{1}{2}$}
    (2.8, -0.1) node{$\frac{1}{2}$}
    (3.1, 1.0) node{$\frac{1}{3}$}
    (3.4, 1.8) node{$\frac{1}{12}$}
    (3.1, 2.7) node{$\frac{1}{6}$}
    (1.5, 2.9) node{$\frac{3}{4}$}
    (2.0, 2.0) node{$1$}
    (1.5, 1.6) node{$\frac{2}{3}$}
    (0.8, 1.7) node{$\frac{1}{6}$}
    (1.6, 0.7) node{$\frac{5}{6}$};
    \draw[-stealth] (ghost1) -- (s1);
    \draw[-stealth] (s1) to[out = 0, in = 180] (s2);
    \draw[-stealth] (s1) -- (s3);
    \draw[-stealth] (s2) -- (s3);
    \draw[-stealth] (s2) to[out = -20, in = 110] (s5);
    \draw[-stealth] (s5) to[out = 160, in = -70] (s2);
    \draw[-stealth] (s4) -- (s7);
    \draw[-stealth] (s4) -- (s8);
    \draw[-stealth] (s3) to[out = -30, in = -150] (s4);
    \draw[-stealth] (s4) to[out = 150, in = 30] (s3);
    \draw[-stealth] (s7) to[out = 30, in = -30, looseness = 8] (s7);
    \draw[-stealth] (s8) to[out = 30, in = -30, looseness = 8] (s8);
    \draw[-stealth] (s5)--(s8);
\end{tikzpicture}
\vspace{10pt}
\caption{\DTMC $M_e-S_0$}
\label{fig:DTMC_example-s5s6}
\end{minipage}\hfill
\begin{minipage}[b]{0.49\textwidth}
\centering
\begin{tikzpicture}
    \draw
    (0.3, 1) node(ghost1){}
    (1, 1) node[circled node](s1){$s_1$}
    (3.8, 1) node[circled node](s2){$s_2$}
    (2.4, -0.4) node[circled node](s5){$s_5$}
    (3.8, -0.4) node[circled node](s6){$s_6$}
    (3.8, 2.4) node[circled node](s7){$s_7$}
    (2.4, 1) node[circled node](s8){$s_8$};
    \draw
    (4.7, 2.4) node{$1$}
    (2.4, 1.9) node{$1$}
    (4.0, 0.3) node{$\frac{1}{4}$}
    (3.2, -0.9) node{$\frac{1}{2}$}
    (3.2, -0.3) node{$1$}
    (3.0, 0.7) node{$\frac{1}{4}$}
    (4.0, 1.6) node{$\frac{4}{9}$}
    (1.7, 1.3) node{$\frac{13}{54}$}
    (3.2, 1.3) node{$\frac{2}{9}$}
    (2.2, 2.6) node{$\frac{13}{27}$}
    (2.7, 0.2) node{$\frac{1}{3}$}
    (1.4, 0.2) node{$\frac{5}{18}$};
    \draw[-stealth] (ghost1) -- (s1);
    \draw[-stealth] (s2) -- (s5);
    \draw[-stealth] (s6) -- (s2);
    \draw[-stealth] (s5) to[out = 30, in = 150] (s6);
    \draw[-stealth] (s6) to[out = -150, in = -30] (s5);
    \draw[-stealth] (s6) -- (s8);
    \draw[-stealth] (s7) to[out = 30, in = -30, looseness = 8] (s7);
    \draw[-stealth] (s8) to[out = 120, in = 60, looseness = 8] (s8);
    \draw[-stealth] (s2) -- (s7);
    \draw[-stealth] (s1) to[out = 90, in = 180] (s7);
    \draw[-stealth] (s1) -- (s5);
    \draw[-stealth] (s1) -- (s8);
    \draw[-stealth] (s2) -- (s8);
\end{tikzpicture}
\vspace{-6pt}
\caption{\DTMC $M_e-\set{s_1,s_2,s_3,s_4}$}
\label{fig:DTMC_example-s1s2s3s4}
\end{minipage}
%\vspace{-1em}
\end{figure}

\begin{example}
Consider $M_e$ from \Cref{subfig:DTMC_example} again.
Now following the recursive structure of ``$\ominus$'', one additional path abstraction is computed compared to \Cref{ex:first-method}:
over $S_0=\set{s_5,s_6}$ at the beginning.
Our sequence of subsets of $K$ becomes $(S_0,S_1,S_2)$,
and our sequence of {\DTMC}s becomes $M_e$, $M_e-S_0$, $M_e-(S_0,S_1)=M_e-S_1$\footnote{Note that this holds by \Cref{thm:path-abstr-main-results}, as $S_0\subset S_1$.}, $M_e-(S_1,S_2)$, and $M_e-K$,
depicted in Figs.~\ref{subfig:DTMC_example}, \ref{fig:DTMC_example-s5s6}, \ref{subfig:DTMC_example-s2s5s6}, \ref{subfig:DTMC_example-s2s5s6-s3s4}, and \ref{subfig:DTMC_example-K}, respectively.
\tikzset{circled node/.style={circle,draw, inner sep=0, minimum size=1.5em},
nodestyle/.style={circled node,text height=.8em,text depth=.25em}}
\end{example}
Our formulation of the two methods is slightly more general, and fixes a small mistake, compared to \cite{AJWKB10}.
We provide detailed notes on the differences in \ifarxivelse{\Cref{app:differences-with-AJWKB10}}{App.~C}.

\paragraph{Counterexample refinement.}
%The main implication of \Cref{cor:path-abstr-refinement} for counterexample refinement is as follows.
Suppose we choose a threshold $\lambda\in[0,1]$ and wish to %, for some goal state $s\in S\setminus K$,
determine whether $\P^a(\diamondsuit\,s)\leqslant\lambda$. %, which is a common problem in PCTL Model Checking on {\DTMC}s.
Of course, we can abstract over $K$ straight away % the set of all non-absorbing states
to immediately compute $\P^a(\diamondsuit\,s)=\P_K^a(as)$ and check the threshold.
However, due to \Cref{cor:path-abstr-refinement}, we can now choose \emph{any} sequence of $t\in\NN^+$ subsets $S_1,\ldots,S_t\subseteq K$ of $K$ as interesting as we like, and iteratively compute a sequence of (substochastic) {\DTMC}s according to $M_0:=M$ and $M_i:=M_{i-1}-S_i$ for $i=1,2,\ldots$ ($i\leqslant t$) until we hit a \DTMC $M_i$ that has a path from $a$ to $s$ with probability $>\lambda$.
Then we can go back along the sequence to trace what path set (which is a subset of $aK^*s$) was responsible for exceeding the threshold.

Our contribution here is that $S_1,\ldots,S_t$ does not have to follow the \SCC structure as in \cite{AJWKB10}, but can be \emph{any} finite sequence of subsets of $K$, thus extending the method well beyond the connectivity of the underlying digraph of the \DTMC.

\begin{figure}[t]
        \centering
    \vspace{-1em}
\end{figure}

\begin{example}
Consider $M_e$ from \Cref{subfig:DTMC_example} and reachability constraint $\P^{s_1}(\diamondsuit\,s_7)\leqslant\frac{4}{9}$.
It is not immediate from \Cref{subfig:DTMC_example} that this constraint is violated by $M_e$;
neither is it clear from the {\DTMC}s from Figs.~\ref{fig:DTMC_example-s5s6}, \ref{subfig:DTMC_example-s2s5s6}, and \ref{subfig:DTMC_example-s2s5s6-s3s4}, which are the intermediate abstractions in the recursive \SCC-based method.
Instead of following a sequence of subsets of the state space implicitly generated by any \SCC-based method, we can, for the \DTMC $M_e$ from \Cref{subfig:DTMC_example}, compute the abstraction $M_e-\set{s_1,s_2,s_3,s_4}$ (without having to abstract over $K=\set{s_1,\ldots,s_6}$ itself), as depicted in \Cref{fig:DTMC_example-s1s2s3s4}.
Then, we can see that the reachability constraint $\P^{s_1}(\diamondsuit\,s_7)\leqslant\frac{4}{9}$ for the original \DTMC $M_e$ in \Cref{subfig:DTMC_example} is violated, since the transition $(s_1,s_7)$ in the \DTMC of \Cref{fig:DTMC_example-s1s2s3s4} has probability $\frac{13}{27}>\frac{4}{9}$---without having to abstract over other parts of the \DTMC.
\tikzset{circled node/.style={circle,draw, inner sep=0, minimum size=1.5em},
nodestyle/.style={circled node,text height=.8em,text depth=.25em}}
\end{example}

\section{Numerical Methods}
\label{sec:num-methods}

Let us briefly recall some numerical methodology typically used to compute reachability probabilities for DTMC model checking.
Good reference material includes~\cite[Sec.~10.1.1]{BK08} and~\cite{GM02}.
Our goal is to design a reference implementation to compute path abstractions, more abstract than that of~\cite[Algorithm~3]{AJWKB10} and  executable.
Throughout, given $\P\in\PP(S)$,
let $\P_2\in\RR^{S\times S}$ denote the transition probability matrix corresponding to $\P$---i.e., $\P_2(s,t)=\P(st)$ for all $s,t\in S$.

\begin{theorem}
\label{thm:num-methods}
    %Let $S$ be a finite set, and let $\P\in\PP(S)$.
    %Then, the following hold.
    %\begin{enumerate}
        (1)~The matrix $\P_2$ is substochastic, and
        (2)~if $S_1\subseteq S$, $i\in\NN^+$, and $s,r\in S_1$, then $\P(sS_1^{i-1}r)=\P_2(S_1)^i(s,r)$.
    %\end{enumerate}
\end{theorem}

\begin{proof}[sketch]
Part 1: by \Cref{def:PMF}--2.
Part 2: by induction on~$i$.
The full proof is included in \ifarxivelse{\Cref{app:Proofs5}}{App.~D}.
\end{proof}
Now, we transform the above results to reachability probabilities for DTMC.
\begin{theorem}
\label{thm:num-methods-path-abstr}
    %Let $S$ be a finite set, let $S_1\subsetneq S$, let $a\in S$ and let $M=(\P,a)\in\M(S,a)$.
    Let $S_1\subsetneq S$ and set
    \begin{align*}
        \I&:=S_1\setminus(S_1)_0^M, &
        \O&:=\set{t\in S\setminus S_1\mid\P(S_1t)>0},\\
        \U&:=\set{r\in S_1\mid\P(rS_1^*\O)>0}, &
        \U_1&:=\set{r\in S_1\mid\P(r\O)>0}.
    \end{align*}
    Suppose that $\I\cap\U$ and $\O$ are nonempty. Let $s\in\I\cap\U$ and $t\in\O$.
    Then, we have
    \begin{equation}
    \P_{S_1}^a(st)=\P(st+S_1)=\P(sS_1^*t)=Q(s,\U_1)\P_2(\U_1,t)
    \end{equation}
    where $Q(\I\cap\U,\U_1)$ is extracted from the unique solution $Q\in\RR^{\U\times\U_1}$ of the system of equations
    $(\1-\P_2(\U))\cdot Q=\1(\U,\U_1)$.
\end{theorem}

\begin{proof}[sketch]
By showing that $\sum_{i=0}^\infty\P_2(\U)^i$ converges by noting that the spectral radius of $\P_2(\U)$ is smaller than 1, which follows from the definition of $\U$, and combining this with \Cref{def:path-abstr}.
The full proof is included in \ifarxivelse{\Cref{app:Proofs5}}{App.~D}.
\end{proof}
Now, as the computations of the sets $\I$, $\O$ and $\U_1$ are obvious but the computation of $\U$ is not, let us present a reachability algorithm to compute $\U$.

\begin{lemma}
\label{lem:num-methods-path-abstr-compute-U}
    Consider \Cref{thm:num-methods-path-abstr}.
    Then, $\U$ can be computed as follows:
    \begin{enumerate}
        \item Set $\U_0:=\O$, set $\V_0:=\U_0$, and set $i:=0$.
        \item While $\U_i\neq\varnothing$,
        set $\U_{i+1}:=\set{r\in S_1\setminus\V_i\mid\P(r\U_i)>0}$, set $\V_{i+1}:=\V_i\uplus\U_{i+1}$,
        and set $i:=i+1$. %End While.
        \item Output $\V_i\setminus\O$, which is precisely $\U$.
    \end{enumerate}
\end{lemma}

\begin{proof}[sketch]
Straightforward.
The full proof is included in \ifarxivelse{\Cref{app:Proofs5}}{App.~D}.
\end{proof}
The observation now is that \Cref{thm:num-methods-path-abstr} and \Cref{lem:num-methods-path-abstr-compute-U} provide a full recipe for a numerical implementation of path abstraction:
We compute the sets $\I$, $\O$, $\U_1$, and $\U$,
and solve $(\1-\P_2(\U))\cdot Q=\1(\U,\U_1)$ for $Q$.
Now, we build the matrix $Y\in\RR^{S\times S}$ using matrix comprehension as follows.
We iterate over $(s,t)\in S\times S$, and we set, using the ``if-then-else blocks'' $[\cdot,\cdot,\cdot]$,
\begin{equation}
\label{eq:compute-Y}
\begin{split}
    Y(s,t):=[&s\in S\setminus S_1\land t\in(S\setminus S_1)\uplus\I,\P_2(s,t),\\&[s\in\I\cap\U\land t\in\O,Q(s,\U_1)\P_2(\U_1,t),0]],
\end{split}
\end{equation}
and then we have determined $(\P_{S_1}^a)_2=Y$ fully.
Indeed:
Whenever $s,t\in S_1$ then $st+S_1=\varnothing$;
whenever $s\in S_1\setminus\I$ or $t\in S_1\setminus\I$ then $\P_{S_1}^a(st)=0$;
whenever $s,t\in S\setminus S_1$ then we keep the original probability value of the transition $(s,t)$;
whenever $s\in S\setminus S_1$ and $t\in\I$ then we have $\P_{S_1}^a(st)=\P(stS_1^*)=\P(st)$ as well as $(stS_1^*)^\leqslant=\set{st}$;
whenever $s\in\I\setminus\U$ and $t\in\O$ then $\P_{S_1}^a(st)=0$ because $\P(sS_1^*t)=0$ in this case (by definition of $\U$);
and whenever $s\in\I\cap\U$ and $t\in\O$ then we have $\P_{S_1}^a(st)=Q(s,\U_1)\P_2(\U_1,t)$ by \Cref{thm:num-methods-path-abstr}.

\section{Reference Implementation}
\label{sec:ref_impl}

\lstdefinelanguage{gp}
{
    morecomment=[l]{\\},
    morecomment=[s]{/*}{*/},
    morekeywords={if,while,prod,sum, concat,for, vecmin,factorint, matsize,gcd, vector, break, matrix, matid, matsolve, vecextract, vecsum, pathAbstr}
}
\definecolor{lightgray}{rgb}{0.95,0.95,0.95}
\definecolor{darkgray}{rgb}{0.3,0.3,0.3}
\lstset{
    classoffset=0,
    language={gp},
    breaklines=true,
    basicstyle=\fontsize{8.5pt}{12pt}\ttfamily,
    linewidth=\textwidth,
    backgroundcolor=\color{white},
    captionpos=b, % Position of the Caption (t for top, b for bottom)
    extendedchars=true, % Allows 256 instead of 128 ASCII characters
    tabsize=2, % number of spaces indented when discovering a tab
    columns=fixed, % make all characters equal width
    xleftmargin=0.0\textwidth,
    xrightmargin=0.0\textwidth,
    keepspaces=true, % does not ignore spaces to fit width, convert tabs to spaces
    showstringspaces=false, % lets spaces in strings appear as real spaces
    breaklines=true, % wrap lines if they don't fit
    frame=rl  bt, % draw a frame at the top, right, left and bottom of the listing
    frameround=tttt;
    framesep=4pt, % round corners
    keywordstyle=\color{black}\bfseries,%frameround=tttt, % make the frame round at all four corners
    framesep=4pt, % quarter circle size of the round corners
    numbers=left, % show line numbers at the left
    numberstyle=\tiny\ttfamily, % style of the line numbers
    commentstyle=\color{gray}, % style of comments
    % style of keywords
    stringstyle=\color{red}, % style of strings
    identifierstyle=\color{black},
}
\counterwithout{lstlisting}{section}
\begin{lstlisting}[float=t, abovecaptionskip=8pt, label={listing:ref-impl}, caption={Computing $M-K$ where $M\in\M(S,a)$ and $K\subseteq S$}]
  pathAbstr = M -> (K ->  \
    X = M[1]; \
    a = M[2]; \
    n = matsize(X)[1]; \
    II = vector(n, i, K[i] && (i == a || \
      vecsum(vector(n, j, !K[j] * X[j, i])) != 0)); \
    OO = vector(n, i, !K[i] && \
      vecsum(vector(n, j, K[j] * X[j, i])) != 0); \
    Uvec = [OO]; \
    Vvec = [OO]; \
    i = 1; \
    while (Uvec[i] != vector(n), \
      Uvec = concat(Uvec, [vector(n, j, K[j] && !Vvec[i][j] && \
        vecsum(vector(n, k, Uvec[i][k] * X[j, k])) != 0)]); \
      Vvec = concat(Vvec, \
        [vector(n, j, Vvec[i][j] || Uvec[i+1][j])]); \
      i++; \
    ); \
    U1 = Uvec[2]; \
    U = vector(n, j, Vvec[i][j] && !OO[j]); \
    XU = matrix(n, n, i, j, (U[i] && U[j]) * X[i, j]); \
    Q = matsolve(matid(n) - XU, \
      matrix(n, n, i, j, i == j && U[i] && U1[j])); \
    Y = matrix(n, n, i, j, if(!K[i] && (!K[j] || II[j]), \
      X[i, j], if(II[i] && U[i] && OO[j], Q[i,] * X[, j], 0))); \
    [Y, a]; \
  );
\end{lstlisting}

Now, we are in position to apply the numerical methodology to compute path abstractions of (substochastic) \DTMC.
We give a high-level reference implementation in a semi-functional setting, as close as possible to the mathematical reasoning of the recipe provided at the end of \Cref{sec:num-methods}.
We use the \lang{PARI/GP} computer algebra system~\cite{PARI2}, and write a \lang{GP} script that can be interpreted by the \lang{GP} interpreter.
Here, {\DTMC}s are represented as pairs whose first element is a square matrix carrying the probabilities and whose second element is the initial state.
The state space is implicitly understood as the set $[n]=\set{1,\ldots,n}$, where $n$ is the dimension of the matrix.
Computing $M_e-S_1$ (\Cref{subfig:DTMC_example-s2s5s6}) from $M_e$ (\Cref{subfig:DTMC_example}), for example, can be done with\\[2pt]
\centerline{\ttfamily {\color{black}\bfseries pathAbstr}([X, 1])([0,1,0,0,1,1,0,0])}\\[2pt]
where {\ttfamily X} is the $8\times 8$ matrix carrying the probabilities of~$M_e$.
The output will be {\ttfamily [Y, 1]}, where {\ttfamily Y} is the matrix carrying the probabilities of the abstracted \DTMC $M_e-S_1$.
We refer to \ifarxivelse{\Cref{app:ref-impl-analysis}}{App.~E} for a full analysis of \Cref{listing:ref-impl}.

\section{Conclusion and Future Work}
\label{sec:Conclusion}

We have recast the path abstraction approach to DTMC model checking of \'{A}brah\'{a}m \etal~\cite{AJWKB10} in terms of the free monoid on the DTMC's state space.
In this setting, the procedure can be expressed elegantly and concisely; in proofs of similar qualities, we have shown that path abstraction can be applied to any number of subsets of the final abstraction set in any order.
This generalizes the approach of~\cite{AJWKB10}, and makes counterexample refinement much more flexible.

We note that path abstraction is similar in spirit, but effectively very different, from the \emph{state elimination} method used for parametric~\cite{Daw04,HHZ11} and exact probabilistic model checking~\cite{HJKQV22}.
In particular, where state elimination may incur an intermediate blowup in transitions, path abstraction never generates more transitions.
We plan to properly study the complexity of path abstraction, and it may then be interesting to investigate if a combination of state elimination and path abstraction could bring together the simplicity of the former with the scalability of the latter.
Just like state elimination, path abstraction should straightforwardly generalize to expect-reward properties, too.

Our reference implementation arguably comes with a relatively high level of trustworthiness given its short distance from the underlying numerical recipes, but is not practically efficient.
It can however be extremely useful as a baseline for testing and comparison for future lower-level, optimized implementations \eg in C or C\# as part of the \tool{mcsta} model checker~\cite{HH15} of the \toolset~\cite{HH14}.

To attain the highest level of trust, we shall formalize the concepts and algorithms of this paper in the interactive theorem prover Isabelle/HOL to (1)~machine-check the proofs and
(2)~derive a verified correct-by-construction \emph{and} fast implementation in LLVM bytecode using the Isabelle Refinement Framework~\cite{LT12,Lam13}, as recently done in probabilistic model checking for SCC finding~\cite{HKL23}, MEC decomposition~\cite{HKL24}, and finally the interval iteration algorithm~\cite{KSAHL25}.

\paragraph{Acknowledgments.}
The authors thank
Milan Lopuhaä-Zwakenberg for inspiring us to use free monoids where suitable---specifically, in the analysis of \emph{repeatedly} applying path abstraction---and use transition probability matrices for the numerical analysis of path abstractions as separate operations,
and Benedikt Peterseim for providing detailed feedback on drafts of this paper.

\bibliographystyle{splncs04}
\bibliography{paper}

\iftoggle{arxiv}{%
\clearpage
\appendix
\crefalias{section}{appendix} % Cref calls it Section otherwise
\crefalias{subsection}{appendix} % Cref calls it Section otherwise
\section{Proofs of Results of \Cref{sec:Background}}
\label{app:Proofs2}

\begin{proof}[of \Cref{lem:sequence-abstr-is-monotonic-wrt-prefix-order}]
    Let $\Sigma$ be a finite set, and let $\Sigma_1\subseteq\Sigma$.
    Let $x\in\Sigma^+$ (the case $x=\varepsilon$ is clear).
    To show that $x'-\Sigma_1\leqslant x-\Sigma_1$ for all $x'\leqslant x$, it suffices to prove that $x'-\Sigma_1\leqslant x-\Sigma_1$ for the maximal strict prefix $x'$ of $x$, as the general statement follows inductively from the transitivity of the prefix order.
    Thus, write $x=x'a$ where $a\in\Sigma$.
    Then, we obtain
    \begin{align*}
        x-\Sigma_1&=(x'a)-\Sigma_1
        \\&=[a\in\Sigma_1,[x'\in\Sigma^*\Sigma_1,x'-\Sigma_1,(x'-\Sigma_1)a],(x'-\Sigma_1)a]
        \\&\geqslant x'-\Sigma_1.
    \end{align*}
\end{proof}

\begin{proof}[of \Cref{thm:words-combinatorics-main-results}]
    Let $\Sigma$ be a finite set.
    \begin{enumerate}
        \item Let $\Sigma_1\subseteq\Sigma$, let $x,y\in\Sigma^*$, and let $a\in\Sigma$,
        and assume that $xa\notin\Sigma^*\Sigma_1^2$.
        Since $xa\notin\Sigma^*\Sigma_1^2$, we have $xa-\Sigma_1\in\Sigma^*a$.
        Note that $ay-\Sigma_1\in a\Sigma^*$. It follows that $(xa-\Sigma_1)\star(ay-\Sigma_1)$ exists.
        Now, note that $x\notin\Sigma^*\Sigma_1$ if $a\in\Sigma_1$.
        It follows that
        \begin{align*}
            xay-\Sigma_1&=[a\in\Sigma_1,x(ay)-\Sigma_1,(xa)y-\Sigma_1]
            \\&=[a\in\Sigma_1,(x-\Sigma_1)(ay-\Sigma_1),(xa-\Sigma_1)(y-\Sigma_1)]
            \\&=[a\in\Sigma_1,((x-\Sigma_1)a)\star(ay-\Sigma_1),
            (xa-\Sigma_1)\star(a(y-\Sigma_1))]
            \\&=[a\in\Sigma_1,(xa-\Sigma_1)\star(ay-\Sigma_1),(xa-\Sigma_1)\star(ay-\Sigma_1)]
            \\&=(xa-\Sigma_1)\star(ay-\Sigma_1).
        \end{align*}
        \item Let $\Sigma_1\subseteq\Sigma$, let $x,y\in\Sigma^*$, and let $a\in\Sigma$.
        \begin{enumerate}[label=($\subseteq$)]
        \item Let $u\in(xa+\Sigma_1)^\leqslant$ and $v\in(ay+\Sigma_1)^\leqslant$.
        Then, $u-\Sigma_1=xa$ and $v-\Sigma_1=ay$,
        and no $u'<u$ satisfies $u'-\Sigma_1=xa$ and no $v'<v$ satisfies $v'-\Sigma_1=ay$.
        Since no $u'<u$ satisfies $u'-\Sigma_1=xa$, it follows that $u\notin\Sigma^*\Sigma_1^2$.
        It follows that $u\in\Sigma^*a$.
        Since clearly $v\in a\Sigma^*$, it follows that $u\star v$ exists.
        By part 1 of \Cref{thm:words-combinatorics-main-results} we obtain $u\star v-\Sigma_1=(u-\Sigma_1)\star(v-\Sigma_1)=xa\star ay=xay$,
        so $u\star v\in xay+\Sigma_1$.
        Now, if there exists $w'<u\star v$ such that $w'-\Sigma_1=xay$, then $u\star v\in\Sigma^*\Sigma_1^2$ hence $u\in\Sigma^*\Sigma_1^2$ or $v\in\Sigma^*\Sigma_1^2$ contradicting at least one of the two knowns that no $u'<u$ satisfies $u'-\Sigma_1=xa$ and no $v'<v$ satisfies $v'-\Sigma_1=ay$.
        It follows that $u\star v\in(xay+\Sigma_1)^\leqslant$.
        \end{enumerate}
        \begin{enumerate}[label=($\supseteq$)]
        \item Let $z\in(xay+\Sigma_1)^\leqslant$.
        Then, $z-\Sigma_1=xay$, and no $z'<z$ satisfies $z'-\Sigma_1=xay$.
        Now, let $u$ be the smallest prefix of $z$ such that $u-\Sigma_1=xa$. Write $v\in\Sigma^*$ for the unique suffix of $z$ such that $z=u\star v$.
        Since prefixes of $u$ are prefixes of $z$, we have $u\in(xa+\Sigma_1)^\leqslant$.
        Clearly, $u\notin\Sigma^*\Sigma_1^2$.
        Now, we have $xay=z-\Sigma_1=u\star v-\Sigma_1=(u-\Sigma_1)\star(v-\Sigma_1)=xa\star(v-\Sigma_1)$ by part 1 of \Cref{thm:words-combinatorics-main-results},
        so $xay=xa\star(v-\Sigma_1)$ and hence $v-\Sigma_1=ay$.
        Now, if there exists $v'<v$ with $v'-\Sigma_1=ay$, then $u\star v'<z$, but, also, by again part 1 of \Cref{thm:words-combinatorics-main-results},
        $(u\star v')-\Sigma_1=(u-\Sigma_1)\star(v'-\Sigma_1)=xa\star ay=xay$ so $z'=u\star v'$ satisfies $z'<z$ and $z'-\Sigma_1=xay$, which contradicts $z\in(xay+\Sigma_1)^\leqslant$.
        Hence, no such $v'<v$ exists, and it follows that $v\in(ay+\Sigma_1)^\leqslant$.
        Since $u\in (xa+\Sigma_1)^\leqslant$, it follows that
        $z=u\star v\in(xa+\Sigma_1)^\leqslant\star(ay+\Sigma_1)^\leqslant$.
        \end{enumerate}
        \item Let $\Sigma_1\subseteq\Sigma_2\subseteq\Sigma$, and let $x\in\Sigma^*$.
        It is clear that $x-(\Sigma_1,\Sigma_2)=x-\Sigma_2$ if we can show that each letter $a\in\Sigma_2^+$ of $x$ when we consider $x$ as a word of minimal length over the alphabet $\Sigma_2^+\uplus(\Sigma\setminus\Sigma_2)^+$ has the property $a-(\Sigma_1,\Sigma_2)=a-\Sigma_2$,
        since $\Sigma_1\subseteq\Sigma_2$.
        To this end, let $a\in\Sigma_2^+$.
        Since $\Sigma_1\subseteq\Sigma_2$, it must be that
        $a-\Sigma_1\in a_1\Sigma_2^*$ and hence $a-(\Sigma_1,\Sigma_2)=a_1$.
        But we also have $a-\Sigma_2=a_1$ directly.
        Conclude that $a-(\Sigma_1,\Sigma_2)=a-\Sigma_2$ and hence that $x-(\Sigma_1,\Sigma_2)=x-\Sigma_2$.
        \item Let $\Sigma_1\subseteq\Sigma_2\subseteq\Sigma$, and let $x\in\Sigma^*$.
        First, we show that the union written as is is indeed a disjoint union.
        To this end, let $y,z\in(x+\Sigma_2)^\leqslant$ satisfy $y-\Sigma_1=y$ and $z-\Sigma_1=z$.
        Suppose that there exists $u\in(y+\Sigma_1)^\leqslant\cap(z+\Sigma_1)^\leqslant$.
        Then, we have $u-\Sigma_1=y$ and $u-\Sigma_1=z$ so $y=z$.
        \begin{enumerate}[label=($\subseteq$)]
        \item Let $y\in(x+\Sigma_2)^\leqslant$.
        By part 3 of \Cref{thm:words-combinatorics-main-results},
        we have $y-(\Sigma_1,\Sigma_2)=y-\Sigma_2=x$
        so $y-\Sigma_1\in x+\Sigma_2$.
        Now, if $w<y-\Sigma_1$ satisfies $w-\Sigma_2=x$,
        then there exists $y'<y$ with $w=y'-\Sigma_1$.
        By part 3 of \Cref{thm:words-combinatorics-main-results}, however, we obtain
        $y'-\Sigma_2=y'-(\Sigma_1,\Sigma_2)=w-\Sigma_2=x$
        so $y'\in x+\Sigma_2$ which contradicts $y\in(x+\Sigma_2)^\leqslant$.
        It follows that $y-\Sigma_1\in(x+\Sigma_2)^\leqslant$.
        Furthermore, we have $(y-\Sigma_1)-\Sigma_1=y-\Sigma_1$ (by part 3 of \Cref{thm:words-combinatorics-main-results})
        and so it follows that $((y-\Sigma_1)+\Sigma_1)^\leqslant$ is one of the sets appearing in the disjoint union.
        To show that $y$ is in the disjoint union it then suffices to show that $y\in((y-\Sigma_1)+\Sigma_1)^\leqslant$,
        which is equivalent to $y\notin\Sigma^*\Sigma_1^2$.
        But this is clear, since \emph{if} $y\in\Sigma^*\Sigma_1^2$, then $y\in\Sigma^*\Sigma_2^2$ as well (as $\Sigma_1\subseteq\Sigma_2$)
        so there exists a strict prefix of $y$ in $x+\Sigma_2$ which contradicts the fact that $y\in(x+\Sigma_2)^\leqslant$.
        It follows that $y\in((y-\Sigma_1)+\Sigma_1)^\leqslant$ and we are done.
        \end{enumerate}
        \begin{enumerate}[label=($\supseteq$)]
            \item Let $y\in(x+\Sigma_2)^\leqslant$ be such that $y-\Sigma_1=y$, and let $z\in(y+\Sigma_1)^\leqslant$.
            By part 3 of \Cref{thm:words-combinatorics-main-results} we have
            $z-\Sigma_2=(z-\Sigma_1)-\Sigma_2=y-\Sigma_2=x$
            so $z\in x+\Sigma_2$.
            Now, suppose that there exists $z'<z$ with
            $z'-\Sigma_2=x$.
            By \Cref{lem:sequence-abstr-is-monotonic-wrt-prefix-order},
            we have $z'-\Sigma_1\leqslant z-\Sigma_1$.
            By part 3 of \Cref{thm:words-combinatorics-main-results},
            we have $(z'-\Sigma_1)-\Sigma_2=z'-\Sigma_2=x$.
            It follows that $z'-\Sigma_1\in x+\Sigma_2$.
            Since $z'-\Sigma_1\leqslant z-\Sigma_1=y\in(x+\Sigma_2)^\leqslant$ and $z'-\Sigma_1\in x+\Sigma_2$,
            we must in fact have the equality $z'-\Sigma_1=z-\Sigma_1$.
            But we have $z'<z$ so $z\in\Sigma^*\Sigma_1^2$,
            which, in turn, contradicts $z\in(y+\Sigma_1)^\leqslant$.
            It follows that such $z'<z$ does not exist, and it follows that $z\in(x+\Sigma_2)^\leqslant$.
        \end{enumerate}
    \end{enumerate}
\end{proof}

\begin{proof}[of \Cref{thm:PMF-main-results}]
    Let $S$ be a finite set, and let $\P\in\PP(S)$.
    \begin{enumerate}
        \item Let $R\subseteq S^+$.
        Due to~\cite[Theorem 0.0.2]{T11}, the order of summation of countably many nonnegative real numbers is irrelevant to obtain values in $[0,\infty]$.
        Furthermore, any subset $A\subseteq R^\leqslant$ of $R^\leqslant$ has the property that $A^\leqslant=A$. It follows that
        \begin{align*}
            \P(R)&=\sum_{x\in R^\leqslant}\P(x)
            \\&=\sum_{k\in\NN^+}\sum_{x\in R^\leqslant\cap S^k}\P(x)
            \\&=\lim_{k\to\infty}\sum_{j=0}^k\sum_{x\in R^\leqslant\cap S^j}\P(x)
           \\&=\lim_{k\to\infty}\sum_{x\in R^\leqslant\cap S^{\leqslant k}}\P(x)
            \\&=\lim_{k\to\infty}\sum_{x\in(R^\leqslant\cap S^{\leqslant k})^\leqslant}\P(x)
            \\&=\lim_{k\to\infty}\P(R^\leqslant\cap S^{\leqslant k}).
        \end{align*}
        \item Let $s\in S$, and let $R\subset sS^*$ be finite.
        The result is clear for the case $R=\varnothing$. Thus, suppose that $R\neq\varnothing$.
        We proceed by induction on the \emph{horizon} $h(R^\leqslant):=\max\set{m\in\NN\mid R^\leqslant\cap S^m\neq\varnothing}$ of $R^\leqslant$, which always exists as a positive integer as $R$ is finite and $R$ contains a nonempty sequence.
        Throughout, note that for any $T\subseteq S^+$ we have $\P(T)=\P(T^\leqslant)$ as $(T^\leqslant)^\leqslant=T^\leqslant$.
        We incorporate two base cases.
        \begin{itemize}
            \item \emph{Case }$h(R^\leqslant)=1$.
            Then, $R^\leqslant=\set{s}$ so $\P(R)=1\leqslant1$. \checkmark
            \item \emph{Case }$h(R^\leqslant)=2$.
            Then, $R^\leqslant\subseteq sS$, because if $s\in R^\leqslant$ then we would have $h(R^\leqslant)=1$.
            By part 4 of \Cref{thm:PMF-main-results} and by \Cref{def:PMF}--2, we obtain $\P(R)=\P(R^\leqslant)\leqslant\P(sS)=\sum_{t\in S}\P(st)\leqslant1$. \checkmark
        \end{itemize}
        Now, we proceed to the induction step.
        \begin{itemize}
            \item Let $H\geqslant2$ be an integer, and suppose that $\P(T)\leqslant1$ holds for \emph{all} finite $T\subset S^+$ with $h(T^\leqslant)\leqslant H$ and $T\subset s_TS^*$ where $s_T\in S$.
            Now, let $s\in S$, and suppose $R\subset sS^*$ is finite with $h(R^\leqslant)=H+1$.
            Then, we can decompose $R^\leqslant$ into a finite disjoint union $R^\leqslant=\biguplus_{i=1}^tsR_i$ where $t\in\NN^+$ and each $R_i\subset S^*$ is finite with $h(R_i^\leqslant)\leqslant H$ and satisfies $R_i\subset s_iS^*$ for some $s_i\in S$.
            Since $sR_i\subseteq R^\leqslant$ for each $i$ we have $(sR_i)^\leqslant=sR_i$ and hence $R_i^\leqslant=R_i$.
            Carrying on, we arrive at
            \begin{align*}
                \P(R)&=\sum_{x\in R^\leqslant}\P(x)
                \\&=\sum_{i=1}^t\sum_{x\in sR_i}\P(x)
                \\&=\sum_{i=1}^t\sum_{x\in R_i}\P(sx)
                \\&=\sum_{i=1}^t\sum_{x\in R_i^\leqslant}\P(sx)
                \\&=\sum_{i=1}^t\sum_{x\in R_i^\leqslant}\P(ss_i\star x)\\&=\sum_{i=1}^t\sum_{x\in R_i^\leqslant}\P(ss_i)\P(x)
                \\&=\sum_{i=1}^t\P(ss_i)\sum_{x\in R_i^\leqslant}\P(x)
                \\&=\sum_{i=1}^t\P(ss_i)\P(R_i)
                \\&\leqslant\sum_{i=1}^t\P(ss_i)
                \leqslant 1,
            \end{align*}
            where the one-to-final inequality in the above follows from $\P(R_i)\leqslant1$ for all $i=1,\ldots,t$, which in turn follows from the induction hypothesis, and the final equality in the above follows from \Cref{def:PMF}-2. \checkmark
        \end{itemize}
        \item Let $s\in S$, and let $R\subseteq sS^*$, with $R$ not necessarily finite.
        If $k\in\NN$ then $R^\leqslant\cap S^{\leqslant k}\subset sS^*$ is finite, so from part 2 of \Cref{thm:PMF-main-results} it follows that $\P(R^\leqslant\cap S^{\leqslant k})\leqslant1$.
        From part 1 of \Cref{thm:PMF-main-results}, we then obtain
        $\P(R)=\lim_{k\to\infty}\P(R^\leqslant\cap S^{\leqslant k})\leqslant1$.
        \item Let $R\subseteq T\subseteq S^+$.
        Since we do not necessarily have $R^\leqslant\subseteq T^\leqslant$---a counterexample over the Greek alphabet would be $R=\set{\alpha\beta}$, $T=\set{\alpha\beta,\alpha}$ for which we have $R^\leqslant=\set{\alpha\beta}$ but $T^\leqslant=\set{\alpha}$---we construct a function $f:R^\leqslant\to T^\leqslant$ that maps subsets of $R^\leqslant$ to their least common prefix in $T$ (which is in $T^\leqslant$).
        To this end, for $x\in R^\leqslant$, observe that $x\in T$, so there exists a unique nonempty prefix $f(x)\leqslant x$ of $x$ in $T^\leqslant$.
        The partition of $R^\leqslant$ is now given as the disjoint union of the pre-images $f^{-1}(y)=\set{x\in R^\leqslant\mid f(x)=y}$ of $y\in T^\leqslant$ under $f$, and we note that each such set has the property that $\P(f^{-1}(y))\leqslant\P(y)$, as all $x\in f^{-1}(y)$ have $y$ as a common nonempty prefix.
        Since $(f^{-1}(y))^\leqslant=f^{-1}(y)$ for all $y\in T^\leqslant$ as $f^{-1}(y)\subseteq R^\leqslant$ for all $y\in T^\leqslant$, it follows that $\P(R)=\P(R^\leqslant)=\sum_{y\in T^\leqslant}\P(f^{-1}(y))\leqslant\sum_{y\in T^\leqslant}\P(y)=\P(T)$ (where for every occurence of $y\in T^\leqslant$ for which there is no $x\in R^\leqslant$ with $f(x)=y$, we simply have $f^{-1}(y)=\varnothing$ and $\P(f^{-1}(y))=0$ as $\P(\varnothing)=0$).
    \end{enumerate}
\end{proof}

\section{Proofs of Results of \Cref{sec:path-abstr}}
\label{app:Proofs3}

\begin{proof}[of \Cref{lem:PKa-in-PMF(S)}]
    Let $S$ be a finite set, let $a\in S$, let $M=(\P,a)\in\M(S,a)$,
    and let $S_1\subseteq S$.
    The fact that $\P_{S_1}^a(x)\leqslant1$ for all $x\in S^+$ is immediate from the fact that $\P_{S_1}^a(x)$ is one if $|x|=1$, is zero if $|x|>1$ and $x$ passes through $(S_1)_0^M$, and is the $\P$-probability mass of the set $x+S_1\subseteq x_1S^*$ otherwise.
    In the latter case, we have $\P_{S_1}^a(x)=\P(x+S_1)\leqslant\P(x_1S^*)\leqslant1$ where the respective inequalities follow from parts 4 and 3 of \Cref{thm:PMF-main-results}.
    In all cases, $\P_{S_1}^a(x)\leqslant1$.
    Next, we prove that $\P_{S_1}^a$ satisfies the three defining properties of $\PP(S)$ as given in \Cref{def:PMF}.
    \begin{enumerate}
        \item By \Cref{def:path-abstr}, we have $\P_{S_1}^a(x)=1$ for all $x\in S^1$.
        \item Let $s\in S$.
        Clearly, $\P_{S_1}^a(st)\leqslant\P(st+S_1)$ for all $t\in S$.
        Now, if $s\notin S_1$, then we have $\sum_{t\in S}\P_{S_1}^a(st)\leqslant\sum_{t\in S}\P(st+S_1)
        =\sum_{t\in S_1}\P(stS_1^*)+\sum_{t\in S\setminus S_1}\P(st)=\sum_{t\in S}\P(st)\leqslant1$
        since $\P(stS_1^*)=\P(st)$ for all $t\in S_1$ (as $(stS_1^*)^\leqslant=\set{st}$ for all $t\in S_1$).
        If $s\in S_1$, then clearly for all $t\in S_1$ we have $st+S_1=\varnothing$ and for all $t\in S\setminus S_1$ we have $st+S_1=sS_1^*t$.
        It follows that
        $\sum_{t\in S}\P_{S_1}^a(st)\leqslant\sum_{t\in S\setminus S_1}\P(sS_1^*t)$.
        The observation now is that for all $t\in S\setminus S_1$, we have $(sS_1^*t)^\leqslant=sS_1^*t$.
        Since $sS_1^*t$ and $sS_1^*r$ are disjoint for distinct $t,r\in S\setminus S_1$,
        we obtain $(\uplus_{t\in S\setminus S_1}sS_1^*t)^\leqslant=\uplus_{t\in S\setminus S_1}sS_1^*t$.
        It follows that $\sum_{t\in S\setminus S_1}\P(sS_1^*t)=\P(\uplus_{t\in S\setminus S_1}sS_1^*t)\leqslant1$,
        where the inequality follows from part 3 of \Cref{thm:PMF-main-results} as $\uplus_{t\in S\setminus S_1}sS_1^*t\subseteq sS^*$.
        Conclude that, again, $\sum_{t\in S}\P_{S_1}^a(st)\leqslant1$.
        \item Let $x,y\in S^*$ and $s\in S$.
        If $xs$ or $sy$ passes through $(S_1)_0^M$ then so does $xsy$, so $\P_{S_1}^a(xs)\P_{S_1}^a(sy)$ and $\P_{S_1}^a(xsy)$ are both zero (if $x\neq\varepsilon\lor y\neq\varepsilon$) or both $1$ (if $x=\varepsilon\land y=\varepsilon$) hence equal.
        Thus, suppose that $xs$ and $sy$ do \emph{not} pass through $(S_1)_0^M$. Hence, so does $xsy$.
        Now, if $x=\varepsilon\lor y=\varepsilon$ then it is immediate that  $\P_{S_1}^a(xs)\P_{S_1}^a(sy)=\P_{S_1}^a(xsy)$.
        Thus, suppose now that $x$ and $y$ are both nonempty.
        To prove that $\P_{S_1}^a(xs)\P_{S_1}^a(sy)=\P_{S_1}^a(xsy)$ in this case, we thus must show that $\P(xs+S_1)\P(sy+S_1)=\P(xsy+S_1)$.
        Note that $xsy=xs\star sy$.
        By part 2 of \Cref{thm:words-combinatorics-main-results}, we then obtain
        \begin{align*}
            \P(xsy+S_1)&=\sum_{z\in(xsy+S_1)^\leqslant}\P(z)
            \\&=\sum_{u\in(xs+S_1)^\leqslant}\sum_{v\in(sy+S_1)^\leqslant}\P(u\star v)
            \\&=\sum_{u\in(xs+S_1)^\leqslant}\sum_{v\in(sy+S_1)^\leqslant}\P(u)\P(v)
            \\&=\left(\sum_{u\in(xs+S_1)^\leqslant}\P(u)\right)
            \left(\sum_{v\in(sy+S_1)^\leqslant}\P(v)\right)
            \\&=\P(xs+S_1)\P(sy+S_1),
        \end{align*}
        completing the proof.
    \end{enumerate}
\end{proof}

\begin{proof}[of \Cref{thm:path-abstr-main-results}]
    Let $S$ be a finite set, and let $a\in S$.
        Let $S_1\subseteq S_2\subseteq S$, and let $M=(\P,a)\in\M(S,a)$.
        We first show that $(S_1)_0^M\subseteq(S_2)_0^M$.
        Let $s\in(S_1)_0^M$.
        Then, $s\in S_1\setminus\set{a}$ and $\P((S\setminus S_1)s)=0$.
        Since $s\in S_1\setminus\set{a}$ and $S_1\subseteq S_2$, we have $s\in S_2\setminus\set{a}$ as well.
        Since $S_1\subseteq S_2$, we have $S\setminus S_2\subseteq S\setminus S_1$,
        so $(S\setminus S_2)s\subseteq(S\setminus S_1)s$, so, by part 4 of \Cref{thm:PMF-main-results},
        we obtain $\P((S\setminus S_2)s)\leqslant\P((S\setminus S_1)s)=0$ so $\P((S\setminus S_2)s)=0$ as well.
        It follows that $s\in(S_2)_0^M$.
        \par Now, we show that $(S_2)_0^M=(S_2)_0^{M-S_1}$. This is equivalent to the assertion that for every $s\in S_2\setminus\set{a}$, we have $\P((S\setminus S_2)s)=0\Leftrightarrow\P^a_{S_1}((S\setminus S_2)s)=0$.
        This is seen as follows. Let $s\in S_2\setminus\set{a}$.
        If $s\in(S_1)_0^M$ then $\P((S\setminus S_2)s)\leqslant\P((S\setminus S_1)s)=0$ where the inequality follows from \Cref{thm:PMF-main-results}--4 and the equality follows from the definition of $(S_1)_0^M$ (conforming to \Cref{def:path-abstr}), and $\P_{S_1}^a((S\setminus S_2)s)=0$ directly by \Cref{def:path-abstr}, which proves the equivalence in case $s\in(S_1)_0^M$.
        If $s\notin(S_1)_0^M$, then for all $r\in S\setminus S_2$ we have
        $\P_{S_1}^a(rs)=[s\in S_1,\P(rsS_1^*),\P(rs)]=\P(rs)$
        as $(rsS_1^*)^\leqslant=\set{rs}$ if $s\in S_1$, which proves the equivalence in case $s\notin(S_1)_0^M$.
        \par Next, let $x\in S^+$. If $|x|=1$ then clearly
        $(\P_{S_1}^a)_{S_2}^a(x)=1$ and $\P_{S_2}^a(x)=1$ by \Cref{def:path-abstr}
        so $(\P_{S_1}^a)_{S_2}^a(x)=\P_{S_2}^a(x)$. Thus, suppose now that $|x|\geqslant2$. If $x$ passes through $(S_2)_0^M=(S_2)_0^{M-S_1}$ then clearly $(\P_{S_1}^a)_{S_2}^a(x)=0$ and $\P_{S_2}^a(x)=0$.
        \par Thus, assume that $x$ does not pass through $(S_2)_0^M$.
        Since $(S_1)_0^M\subseteq(S_2)_0^M$,
        $x$ does not pass through $(S_1)_0^M$ either.
        We now claim that ($\spadesuit$) $\P(y+S_1)=0$ for all $y\in(x+S_2)^\leqslant$ such that $y-S_1=y$ and $y$ passes through $(S_1)_0^M$.
        Let $y$ be as such. Then, $y$ also passes through $(S_2)_0^M$ at least once.
        Let $i$ be any index such that $y_i\in(S_1)_0^M$ (and hence $y_i\in(S_2)_0^M$).
        Since $y-S_2=x$ does \emph{not} pass through $(S_2)_0^M$, however, it must be that $i>1$ and $y_{i-1}\in S_2$.
        Since $y-S_1=y$ and $y_i\in S_1$, it is impossible that $y_{i-1}\in S_1$, so we have $y_{i-1}\in S_2\setminus S_1$.
        In other words, we have found a factor $y_{i-1}y_i\sqsubseteq y$ of $y$ with $y_{i-1}y_i\in(S_2\setminus S_1)(S_1)_0^M$.
        Now, if $z\in y+S_1$, then $y_{i-1}y_i\sqsubseteq z$ as well as $y_{i-1}\notin S_1$.
        Since $y_{i-1}y_i$ is a factor of $z$ and $y_{i-1}y_i\in(S\setminus S_1)(S_1)_0^M$, it follows by definition of $(S_1)_0^M$ that
        $\P(z)\leqslant\P(y_{i-1}y_i)=0$,
        and hence that $\P(y+S_1)=0$.
        \par Furthermore, if $y\in S^*$ satisfies $y-S_1\neq y$, then if there exists $z\in y+S_1$ then $z-S_1=y$ and $z-S_1=z-(S_1,S_1)=y-S_1\neq y=z-S_1$ by part 3 of \Cref{thm:words-combinatorics-main-results}, which is absurd, and hence $y+S_1=\varnothing$.
        It follows that, ($\clubsuit$) whenever $y\in S^{\geqslant2}$ satisfies $y-S_1\neq y$ we have $\P_{S_1}^a(y)\leqslant\P(y+S_1)=\P(\varnothing)=0$.
        It follows, as $x\in S^{\geqslant2}$ does not pass through $(S_2)_0^M=(S_2)_0^{M-S_1}$, that
        \begin{align*}
            (\P_{S_1}^a)_{S_2}^a(x)
            &=\P_{S_1}^a(x+S_2)
            \\&=\sum_{y\in(x+S_2)^\leqslant}\P_{S_1}^a(y)
            \\&=\sum_{y\in(x+S_2)^\leqslant\mid y-S_1=y}\P_{S_1}^a(y)\quad\text{(by (}\clubsuit\text{))}
            \\&=\sum_{y\in(x+S_2)^\leqslant\mid y-S_1=y}\P(y+S_1)\quad\text{(by (}\spadesuit\text{))}
            \\&=\sum_{y\in(x+S_2)^\leqslant\mid y-S_1=y}\sum_{z\in(y+S_1)^\leqslant}\P(z)
            \\&=\sum_{z\in\biguplus\set{(y+S_1)^\leqslant\mid y\in(x+S_2)^\leqslant\land y-S_1=y}}\P(z)
            \\&=\sum_{z\in(x+S_2)^\leqslant}\P(z)\quad\text{(by part 4 of \Cref{thm:words-combinatorics-main-results})}
            \\&=\P(x+S_2)
            \\&=\P_{S_2}^a(x)\quad\text{(as }x\text{ does not pass through }(S_2)_0^M\text{)},
        \end{align*}
        so for the case that $x\in S^{\geqslant2}$ does \emph{not} pass through $(S_2)_0^M$ we have $(\P_{S_1}^a)_{S_2}^a(x)=\P_{S_2}^a(x)$ as well.
        Conclude that $M-(S_1,S_2)=M-S_2$.
\end{proof}

\begin{proof}[of \Cref{cor:path-abstr-refinement}]
    Let $S$ be a finite set, and let $a\in S$.
    Let $K\subseteq S$, and let $S_1,\ldots,S_t\subseteq K$.
    We prove that $M-(S_1,\ldots,S_t,K)=M-K$ for every $M\in\M(S,a)$ by induction on $t\in\NN^+$.
    For $t=1$, the statement is equivalent to \Cref{thm:path-abstr-main-results}.
    Now, let $t\in\NN^+$, and suppose that $M-(S_1,\ldots,S_t,K)=M-K$ for every $M\in\M(S,a)$ for any choice of $S_1,\ldots,S_t\subseteq K$.
    Now, let $S_1,\ldots,S_{t+1}\subseteq K$, and let $M\in\M(S,a)$.
    Then, we have
    \begin{align*}
        M-(S_1,\ldots,S_{t+1},K)
        &=(M-(S_1,\ldots,S_t))-(S_{t+1},K)
        \\&=(M-(S_1,\ldots,S_t))-K\quad\text{(by the base case)}
        \\&=M-(S_1,\ldots,S_t,K)
        \\&=M-K,
    \end{align*}
    where the final equality in the above follows from the induction hypothesis.
\end{proof}

\section{Differences between \Cref{sec:path-abstr}, \Cref{sec:DTMC-model-checking} and the Methods in \cite{AJWKB10}}
\label{app:differences-with-AJWKB10}

\paragraph{Differences with algorithm specifications in \cite{AJWKB10}.} It deserves to be noted that our reformulation of the methods in \cite{AJWKB10} differ at a few points.
First of all, we note that \cite[Algorithm 2]{AJWKB10} does not require the set $S_1$ over which the inputted \DTMC is recursively abstracted to have the property that $(S_1)_0^M\subsetneq S_1$, of which the analogue in \cite{AJWKB10} is $Inp(M,S_1)\neq\varnothing$ (as $Inp(M,S_1)=S_1\setminus(S_1)_0^M$).
This is not a mistake in itself as in \cite[Algorithm 1]{AJWKB10}, \cite[Algorithm 2]{AJWKB10} is called only for subsets $S_1\subseteq S$ for which $Inp(M,S_1)$ is nonempty, yet this omitted pre-condition in \cite[Algorithm 2]{AJWKB10} might pose problems if not properly addressed and used outside the scope of \cite[Algorithm 1]{AJWKB10}. Indeed: If $S_1\subseteq S$ is strongly connected and non-trivial but also has the property that $(S_1)_0^M=S_1$, then we have $\Comp{M,(S_1)_0^M}\setminus\set{\set{s}\mid s\in S\land\P(ss)=0}=\set{(S_1)_0^M}=\set{S_1}$ so the call of $M\ominus S_1$ will yield the infinite recursion $M\ominus S_1=(M\ominus S_1)-S_1=((M\ominus S_1)-S_1)-S_1=\ldots$ which does not terminate.
\paragraph{Erroneous case of selection of set of states for path abstraction.} On line (1) of \cite[Algorithm 1]{AJWKB10}, $\Comp{M,L}$ is computed where $L=\set{s\in S\mid\P(ss)>0}$, which is erroneous in the following sense. Whether or not we omit singletons of states $s$ that satisfy $\P(ss)=0$ is irrelevant, but \emph{including} the absorbing states has the effect that in the path abstraction their self-loops are removed, because they then become their own respective SCCs over which abstractions are performed. This does not have the consequence that reachability probabilities are lost---we still have the property that the only non-loop transitions (of positive probability) in the final abstracted \DTMC are precisely the transitions $a\rightarrow s$ with $s\in S\setminus K$ absorbing, carrying precisely the respective probabilities $\P_K^a(as)=\P(aK^*s)=\P^a(\diamondsuit\,s)$---but the self-loops of the absorbing states then are removed as well, which, inferring from \cite[Figure 5]{AJWKB10}, is not the intention of the authors of \cite{AJWKB10}.
So instead, we choose $K=\set{s\in S\mid\P(ss)<1}$ and abstract over (the components in) $\Comp{M,K}$, before abstracting over $K$ fully.
\iffalse
\paragraph{Assumption that the initial state is non-absorbing.} Note that \emph{any} path abstraction of {\DTMC}s over subsets of the set of all non-absorbing states implicitly assumes that the initial state is non-absorbing. This is so obvious that we chose to leave this precondition out as a whole, but still deserves to be noted: If the initial state is absorbing, then we can still abstract over $K$, but the result will be meaningless. Since the authors of \cite{AJWKB10} do not point out this seemingly obvious requirement for {\DTMC}s---in the sense that it was not formally addressed---we choose to omit this formally as well.
\fi

\section{Proofs of Results of \Cref{sec:num-methods}}
\label{app:Proofs5}

\begin{proof}[of \Cref{thm:num-methods}]
    Let $S$ be a finite set, and let $\P\in\PP(S)$.
    \begin{enumerate}
        \item Since $\P_2(s,t)=\P(st)$ for all $s,t\in S$ (by construction of $\P_2$ from $\P$), the fact that $\P_2$ is a substochastic matrix is immediate from the second defining property of $\P$ in \Cref{def:PMF}.
        \item Let $S_1\subseteq S$.
        We prove the result by induction on $i\in\NN^+$.
        For $i=1$ we have, for $s,t\in S_1$,
        $\P(sS_1^{i-1}t)=\P(st)=\P_2(s,t)=\P_2(S_1)^1(s,t)$.
        Now, suppose that for a given $i\in\NN^+$ we have
        $\P(sS_1^{i-1}t)=\P_2(S_1)^i(s,t)$ for all $s,t\in S_1$.
        Let $s,t\in S_1$.
        Note that $(sS_1^it)^\leqslant=sS_1^it$, as all sequences in $sS_1^it$ have equal length, and note that $sS_1^it=\biguplus_{r\in S_1}sS_1^{i-1}rt$.
        It follows that
        \begin{align*}
            \P(sS_1^it)
            &=\sum_{r\in S_1}\sum_{x\in sS_1^{i-1}r}\P(xt)
            \\&=\sum_{r\in S_1}\sum_{x\in sS_1^{i-1}r}\P(x\star rt)
            \\&=\sum_{r\in S_1}\sum_{x\in sS_1^{i-1}r}\P(x)\P(rt)
            \\&=\sum_{r\in S_1}\P(sS_1^{i-1}r)\P(rt)
            \\&=\sum_{r\in S_1}\P_2(S_1)^i(s,r)\P_2(S_1)(r,t)
            \quad\text{(by the induction hypothesis)}
            \\&=\P_2(S_1)^{i+1}(s,t).
        \end{align*}
    \end{enumerate}
\end{proof}

\begin{proof}[of \Cref{thm:num-methods-path-abstr}]
Let $S$ be a finite set, let $S_1\subsetneq S$, let $a\in S$ and let $M=(\P,a)\in\M(S,a)$.
Set $\I$, $\O$, $\U$ and $\U_1$ according to how it is done in \Cref{thm:num-methods-path-abstr}.
Note that the system of equations $(\1-\P_2(\U))\cdot Q=\1(\U,\U_1)$ has a unique solution $Q\in\RR^{\U\times\U_1}$ if and only if $\1-\P_2(\U)$ is invertible.
But clearly $X:=\sum_{i=0}^\infty\P_2(\U)^i$ converges as every $r\in\U$ is part of a path leading out of $\U$ with positive probability,
which implies that the spectral radius of $\P_2(\U)$ is smaller than $1$, so $\P_2(\U)$ does not have $1$ as an eigenvalue, so $\1-\P_2(\U)$ does not have $0$ as an eigenvalue, so $\1-\P_2(\U)$ is invertible and has inverse $X$, so $Q=X(\U,\U_1)$. Note that $\U_1\subseteq\U$.
Furthermore, since $s\in\I=S_1\setminus(S_1)_0^M$, $(S_1)_0^M\subseteq S_1$, and $t\in\O\subseteq S\setminus S_1$, the $2$-state sequence $st$ does not pass through $(S_1)_0^M$ and so $\P_{S_1}^a(st)=\P(st+S_1)$ by \Cref{def:path-abstr}.
Clearly, as $s\in S_1$ but $t\notin S_1$, we have $st+S_1=sS_1^*t$.
Now, for our given $s\in\I\cap\U$ and $t\in\O$, we have
\begin{align*}
    \P(sS_1^*t)&=\P(s\U^*t)
    \\&=\P(st)+\sum_{i=1}^\infty\P(s\U^it)
    \\&=\P(st)+\sum_{i=1}^\infty\P(s\U^{i-1}\U_1t)
    \\&=\P(st)+\sum_{i=1}^\infty\sum_{r\in\U_1}\P(s\U^{i-1}r)\P(rt)
    \\&=\P(st)+\sum_{r\in\U_1}\P(rt)\sum_{i=1}^\infty\P(s\U^{i-1}r)
    \\&=\P(st)+\sum_{r\in\U_1}\P(rt)\sum_{i=1}^\infty\P_2(\U)^i(s,r),
\end{align*}
where the final equality in the above follows from part 2 of \Cref{thm:num-methods} for the set of states $\U$. Now, if $s\in\U_1$ then $\P(st)=\sum_{r\in\U_1}[s=r]\cdot\P(rt)=\sum_{r\in\U_1}\P(rt)\1(\U)(s,r)$.
If $s\notin\U_1$ then $\P(st)=0$ but also $[s=r]=0$ for all $r\in\U_1$ so we again have $\P(st)=\sum_{r\in\U_1}[s=r]\cdot\P(rt)=\sum_{r\in\U_1}\P(rt)\1(\U)(s,r)$.
In either case, since $\P_2(\U)^0=\1(\U)$, we can continue our computation for $\P(sS_1^*t)$ where we left off as follows to deduce that
\begin{align*}
    \P(sS_1^*t)
    &=\sum_{r\in\U_1}\P(rt)\sum_{i=0}^\infty\P_2(\U)^i(s,r)
    \\&=\sum_{r\in\U_1}\P(rt)X(s,r)
    \\&=X(s,\U_1)\P_2(\U_1,t)
    \\&=Q(s,\U_1)\P_2(\U_1,t),
\end{align*}
and so $\P_{S_1}^a(st)=\P(st+S_1)=\P(sS_1^*t)=Q(s,\U_1)\P_2(\U_1,t)$.
\end{proof}

\begin{proof}[of \Cref{lem:num-methods-path-abstr-compute-U}]
    Clearly, each $\U_i$ (with $i>0$) is the set of states in $S_1$ that reach $\O$ in exactly $i$ steps and no less (with only the last step not fully within $S_1$), and each $\V_i$ (with $i>0$) is the set of states in $S_1$ that reach $\O$ in at most $i$ steps (with only the last step not fully within $S_1$).
    Since $S_1$ is finite, there will be an $i\in\NN^+$ for which $\U_i=\varnothing$; at that point, the algorithm terminates, and $\V_i\setminus\O$ will be the set of all states in $S_1$ that reach $\O$ (with only the last step not fully within $S_1$) as there will not be states in $S_1$ that reach $\O$ in more than at least $i$ steps.
\end{proof}

\section{Correctness Analysis of \Cref{listing:ref-impl}}
\label{app:ref-impl-analysis}
%\subsection{Correctness and Completeness Analysis of \Cref{listing:ref-impl}}
Let us give a correctness argument of \Cref{listing:ref-impl}, based upon the recipe formulated at the end of \Cref{sec:num-methods}.
\paragraph{Line 1: }Here, {\color{black}\bfseries\texttt{pathAbstr}} takes a (substochastic) \DTMC $M$ as input, given as a pair $M=(X,a)$ where $X\in\RR^{n\times n}$ is a substochastic matrix and $a\in[n]$ is the initial state. The state space of $M$ is $[n]=\set{1,\ldots,n}$. The output is a function that takes a subset $K$ of $[n]$, represented by a bitstring of length $n$---more generally, any subset of $[n]$ is represented by a bitstring of length $n$---and outputs the abstracted \DTMC $(Y, a)$ where $Y$ is the substochastic matrix corresponding to the abstracted \DTMC.
\paragraph{Lines 2--4: }Here, $X$ and $a$ are extracted from $M$, and the dimension $n$ of $X$ is extracted from $X$.
\paragraph{Lines 5--8: }Here, the sets $\I$ and $\O$ as defined in \Cref{thm:num-methods-path-abstr} for $S_1=K$ are computed. Since \texttt{I} and \texttt{O} happen to be reserved keywords in \lang{PARI/GP} (for the imaginary unit and big-O notation, respectively), we use the literals \texttt{II} and \texttt{OO} instead.
\paragraph{Lines 9--20: }Here, we exactly follow the procedure in \Cref{lem:num-methods-path-abstr-compute-U} to compute $\U_1$ and $\U$.
Note that on line \texttt{12}, we can use {\color{black}\bfseries\texttt{vector}}\texttt{(n)} to make an instance at runtime of a vector of length $n$ containing only zeros, as specifying no values at all in \lang{PARI/GP} is equivalent to setting all entries to zero.
\paragraph{Line 21: }Here, we compute \texttt{XU} to serve as a representation of the matrix $X(\U)$, and we obtain \texttt{XU} from \texttt{X} by simply replacing all entries that are in a non-$\U$ column or a non-$\U$ row with zero.
\paragraph{Line 22--23: }At this point, by \Cref{thm:num-methods-path-abstr}, we have the guarantee that $\1-X(\U)$ is invertible so we can perform Gaussian elimination to compute $Q$ from the system of equations
    $A\cdot Q=B$ where $A=\1-X(\U)$ and $B=\1(\U,\U_1)$.
    In \lang{PARI/GP}, such system $A\cdot Q=B$ with $A$ invertible is solved by \emph{Gaussian elimination} by the invocation {\color{black}\bfseries\texttt{matsolve}}\texttt{(A, B)} to yield \texttt{Q}, with proper dimension inference for \texttt{Q}.
    Here, we use matrix comprehension to construct an instance of $\1(\U,\U_1)$ at runtime to be passed to the {\color{black}\bfseries\texttt{matsolve}} method.
\paragraph{Lines 24--25: }Here we exactly follow the determination of $Y$ using \eqref{eq:compute-Y}.
\paragraph{Lines 26--27: }We output the representation of the abstracted (substochastic) \DTMC $(Y,a)$ and close off the function body of the {\color{black}\bfseries\texttt{pathAbstr}} method.

}{}

\end{document}